\documentclass[11pt,a4paper]{article}
\usepackage[latin1]{inputenc}
\usepackage{fullpage}  
\usepackage{hyperref}

\usepackage[compatibility=false]{caption}
\usepackage{subcaption}
\usepackage{tikz}
\usetikzlibrary{arrows}

\usepackage{amsmath}  
\usepackage{amsfonts}  
\usepackage{amssymb}   
\usepackage{amsthm}    

\newtheorem{theorem}{Theorem}
\newtheorem{proposition}{Proposition}

\newtheorem{lemma}{Lemma}

\newtheorem{definition}{Definition}

\newcommand{\BO}[1]{{O}\mathopen{}\left(#1\right)\mathclose{}}

\newcommand{\BT}[1]{{\Theta}\mathopen{}\left(#1\right)\mathclose{}}
\newcommand{\BOM}[1]{{\Omega}\mathopen{}\left(#1\right)\mathclose{}}

\usepackage{thm-restate}
\newcommand{\A}{\mathcal A}
\newcommand{\B}{\mathcal B}

\title{A Lower Bound Technique for Communication in BSP\footnote{This work was supported,
in part, by MIUR of Italy under project AMANDA, and by University of Padova under projects STPD08JA32,
CPDA121378/12 and CPDA152255/15. A preliminary
version of this work~\cite{BilardiSS12} was presented at the
\emph{18th International European Conference on Parallel and
Distributed Computing (Euro-Par~2012).}}}

\author{Gianfranco Bilardi\thanks{University of Padova. \hbox{E-mail}:~{\tt bilardi@dei.unipd.it\textnormal{,} silvestri@dei.unipd.it}.}
\and Michele Scquizzato\thanks{KTH Royal Institute of Technology. \hbox{E-mail}:~{\tt mscq@kth.se}.}
\and Francesco Silvestri\footnotemark[2]
}

\begin{document}

\maketitle

\begin{abstract}
Communication is a major factor determining the performance of
algorithms on current computing systems; it is therefore valuable to
provide tight lower bounds on the communication complexity of
computations. This paper presents a lower bound technique for the
communication complexity in the bulk-synchronous parallel (BSP) model
of a given class of DAG computations. The derived bound is expressed
in terms of the switching potential of a DAG, that is, the number of
permutations that the DAG can realize when viewed as a switching
network.
The proposed technique yields tight lower bounds for the fast Fourier transform (FFT),
and for any sorting and permutation network. A stronger bound is also derived for
the periodic balanced sorting network, by applying this technique to suitable subnetworks.
Finally, we demonstrate that the switching potential captures communication
requirements even in computational models different from BSP,
such as the I/O model and the LPRAM.
\end{abstract}

\section{Introduction}\label{sec:introduction}

A substantial fraction of the time and energy cost of a parallel
algorithm is due to the exchange of information between processing and
storage elements. As in all endeavors where performance is pursued, it
is important to be able to evaluate the distance from optimality of a
proposed solution.  In this spirit, we consider lower bounds on the
amount of communication that is required to solve some computational
problems on a distributed-memory parallel system. We model the machine
using the standard bulk-synchronous parallel (BSP) model of
computation~\cite{Valiant90}, which consists of a collection of $p$
processors, each equipped with an unbounded private memory and
communicating with each other through a communication network.

We focus on a metric, called \emph{BSP communication complexity} and
denoted $H$, defined as the sum, over all the supersteps of a BSP
algorithm, of the maximum number of messages sent or received by any
processor (a quantity usually referred to as the \emph{degree} of a
superstep). This metric captures a relevant component of the cost of
BSP computations.  We propose the \emph{switching potential} technique
to derive lower bounds on the BSP communication complexity, which is
applicable under a number of assumptions. (a) The computation can be
modeled in terms of a directed acyclic graph (DAG), whose nodes
represent operations (of both input/output and functional type) and
whose arcs represent data dependencies. The resulting lower bound
holds for all BSP evaluations of the given DAG, which vary depending
on the superstep and the processor chosen for the evaluation of an
operation, and the way (routing path and schedule of the message along
such a path) in which a value is sent from the processor that computes
it to a processor that utilizes it. (b) The internal nodes of the DAG
are restricted to have the same number of incoming and outgoing arcs,
so that they can be thought of as switches that can establish any
one-to-one relation between the incoming arcs and the outgoing
arcs. The switching potential is the number of permutations that can
be established, by means of arc-disjoint paths, between the arcs
incident on the input nodes and those incident on the output
nodes. (c) During the BSP computation, each node of the DAG is
evaluated exactly once (\emph{no recomputation}). (d) The number of
input/output nodes that are mapped to the same BSP processor satisfies
a suitable upper bound, essentially ensuring that the computation is
distributed over at least two processors.


We illustrate the versatility of the switching potential technique in
several ways.  We apply it to derive tight lower bounds on the BSP
communication complexity of the fast Fourier transform (FFT), and of
sorting and permutation networks.  We also show how, for some DAGs,
the lower bound on communication can be boosted by composing the
results provided by the technique for suitable parts of the
DAG. Finally, we demonstrate that the switching potential of a DAG
captures communication requirements which can lead to lower bounds
even in models different from BSP.

\subsection{A Perspective on Previous Work}\label{sec:previous}

The impact of communication on performance has been extensively
investigated. Even a concise review of all the relevant work would go
far beyond the scope of this paper. Here, we will simply attempt to
place our work in the broader context and then review more closely the
results more directly comparable to ours.

A first division can be drawn between studies that consider the
communication requirements inherent to computational problems and
studies that consider the communication requirements of specific
algorithms. Examples of results that apply to computational problems
are the crossing-sequence lower bounds on time for Turing
machines~\cite{Hennie65}, Grigoriev's flow lower bounds on space-time
tradeoffs~\cite{Savage95}, and information-flow lower bounds on the
area-time tradeoff in VLSI~\cite{Yao79,Thompson80,BilardiP86}. One
ingredient of some of the lower bounds in our paper is based on the
information flow of the cyclic shift problem, originally studied
in~\cite{Vuillemin83}.

Our switching potential approach targets lower bounds for classes of
implementations of specific algorithms, typically modeled by
\emph{computation DAGs}, where nodes represent input and functional
operations, and arcs represent data dependencies.  A significant
distinction lies in whether the rules defining the class of
implementations allow for \emph{recomputation}, that is, for the
ability to evaluate a given functional node of the DAG multiple
times. A number of efforts have focused on data movement between
levels of the memory hierarchy,
e.g.,~\cite{HongK81,AggarwalV88,Savage95,BilardiPD00,RanjanSZ11} (with
recomputation)
and~\cite{BilardiP99,RanjanSZ11,BallardDHS11,BallardDHS12} (without
recomputation). Other papers have investigated data movement between
processing elements in distributed computing,
e.g.,~\cite{PapadimitriouU87,AggarwalCS90,MacKenzieR98} (with
recomputation),
and~\cite{Goodrich99,BilardiP99,IronyTT04,BallardDHS11,BallardDHS12,ScquizzatoS14}
(without recomputation).  Re-execution of operations is of interest
because it is known that, in some models of computation, it can be
exploited in order to asymptotically improve performance.  For
example, recomputation is known to have the potential to reduce the
space requirements of computations in the context of classical
pebbling games (see, e.g., \cite{Savage98}), to enhance the
performance of simulations among networks (see \cite{KochLMRRS97} and
references therein), to enable area-universal VLSI
computations~\cite{BhattBP08} with constant slowdown, and to reduce
the number of write operations between levels of the memory
hierarchy~\cite{BlellochFGGS16}.  However, as in most of the present
paper, recomputation is often ruled out, as a simplifying assumption
that still affords the development of insights on the problem, when
the general case proves hard to tackle.

The metrics capturing communication requirements are typically quite
sensitive to the underlying model of computation; even within the
same model, several metrics can be meaningful. In this paper, we will
focus on models where the computation is carried out by a set of
processors with local storage, interconnected by a communication
medium. Examples of such models are bounded degree
networks~\cite{Leighton92}, BSP~\cite{Valiant90},
LogP~\cite{CullerKPSSSSE96}. A pioneering paper in this
area~\cite{PapadimitriouU87} introduced two important metrics and
analyzed them for the diamond DAG: the total number $c$ of messages
exchanged among the processors, called \emph{communication}, and the
maximum number $d$ of messages exchanged as part of the evaluation of
some path of the DAG, called \emph{communication delay}.  Both metrics
can be trivially minimized to zero by assigning the entire DAG to just
one processor. However, interesting tradeoffs arise between the two
metrics and the number of steps $t$ of a parallel schedule. For
example, for the diamond DAG with $n$ inputs, $c=\Omega(n^3/t)$ and
$d=\Omega(n^2/t)$, allowing for recomputation.

Another interesting metric is the maximum number ${\hat
  H}_{\mathrm{rc}}$ of messages received by any processor. Clearly,
${\hat H}_{\mathrm{rc}} \geq c/p$, since each message is received by
some processor. Next, we outline a lower bound technique for ${\hat
  H}_{\mathrm{rc}}$, based on the graph-theoretic notion of
dominator. If $W$ and $U$ are sets of nodes of a given DAG, $W$ is
called a \emph {dominator} of $U$ if every path from an input node to
a node of $U$ includes a node of $W$. Let $D(k)$ be the maximum size
of a set $U$ of nodes that has a dominator $W$ of size $k$. In a
seminal paper~\cite{HongK81}, formulating a framework for the study
I/O complexity in hierarchical memories, the number $Q(S)$ of data
transfers between a ``cache'' of size $S$ and ``main memory'' is shown
to satisfy the lower bound $Q(S) \geq \lfloor \nu S/D(2S) \rfloor$,
where $\nu$ is the number of (input and functional) nodes.  The
dominator technique can be adapted to establish that, if the DAG is
evaluated by $p$ processors, each of which initially stores at most
$q$ inputs, then at least one processor receives ${\hat
  H}_{\mathrm{rc}} \geq D^{-1}(\nu/p)-q$ messages.  The argument goes
as follows: (a) a processor $P$ can receive, as part of the input or
from other processors, the value of at most $(q+ {\hat
  H}_{\mathrm{rc}})$ nodes of the DAG; (b) these nodes must dominate
all the nodes whose value is computed by $P$; (c) by definition of
$D$, at most $D(q+ {\hat H}_{\mathrm{rc}})$ DAG nodes are dominated by
any given set of $(q+ {\hat H}_{\mathrm{rc}})$ nodes; (d) at least one
processor must compute no fewer than $\nu/p$ of the $\nu$ nodes in the
DAG. Combining these premises one obtains that $D(q+ {\hat
  H}_{\mathrm{rc}}) \geq \nu/p$, whence the stated bound on ${\hat
  H}_{ \mathrm{rc}}$.  Variants of the outlined argument have been
used, for matrix multiplication, in~\cite{AggarwalCS90,IronyTT04,ScquizzatoS14}.
Dominator-based lower bounds do allow for recomputation.

By definition, the BSP communication complexity $H$ satisfies $H \geq
{\hat H}_{ \mathrm{rc}}$. In particular, $H$ can be larger than ${\hat H}_{
  \mathrm{rc}}$, if different processors receive messages in different
supersteps. This difference can play a role, as in the case of the BSP
computation of an $n$-input radix-two FFT DAG~\cite{CooleyT65}, with
$\nu = n \log(2n)$ nodes,\footnote{In this paper $\log x$ denotes the
  logarithm to the base two.} a case which has motivated our work.
An FFT implementation on BSP is known~\cite{Valiant90} for which
\begin{equation}\label{eq:valiantub}
H = O({\hat H}_{\mathrm{rc}}) = \BO{\frac{n \log n}{p \log(n/p)}},
\end{equation}
for $1 \leq p \leq n/2$, with each processor initially holding $n/p$
(consecutive) inputs.  Valiant's analysis is based on the well known
property that, if $s$ divides $n$ (a power of 2), then the $n$-input
FFT DAG can be viewed as a sequence of $\lceil \log n / \log s \rceil$
stages, each (i) consisting of $n/s$ disjoint $s$-input FFTs (possibly
incomplete in the last stage) and (ii) being connected to the next one
by a permutation. Letting $s=n/p$ and assigning the $p$ ($n/p$)-input
FFTs within a stage to different processors, the resulting BSP
algorithms executes $K = \lceil \log n / \log (n/p) \rceil$
supersteps, each of degree $O(n/p)$, yielding bound~\eqref{eq:valiantub}. 
Making use of the result~\cite{HongK81} that in the FFT DAG $k$ nodes
dominate at most $D(k) \leq 2k \log k$ nodes,\footnote{In the Appendix we
show that this bound can by improved to $D(k) \leq k \log 2k$.}
the dominator technique yields the lower bound
\begin{equation}\label{eq:hk}
H \geq {\hat H}_{\mathrm{rc}} = \BOM{\frac{n \log n}{p \log((n/p) \log n)}},
\end{equation}
assuming that at most $q = \beta (n \log n)/(p \log((n/p) \log n ))$
inputs are initially available to any processor, for a suitably small
constant $\beta$. We observe that the lower bound~\eqref{eq:hk} does
not match the upper bound~\eqref{eq:valiantub}, for either $H$ or
${\hat H}_{\mathrm{rc}}$, when $p = n/2^{o(\log \log n)}$.

An alternate dominator-based lower bound can be obtained when the
local memory of each processor has at most $m$ locations~\cite{BallardDGLOST16}.
During an interval in which it receives $m$ messages, a processor can evaluate
the at most $D(2m)$ nodes dominated by the (at most) $m$ nodes whose
value is received together with the (at most) $m$ nodes whose value is
locally stored at the beginning of the interval. Then, if $\nu$ nodes
are evaluated by $p$ processors, ${\hat H}_{\mathrm{rc}} = \BOM{m
  \lfloor \nu/(pD(2m)) \rfloor}$. For the FFT, $\nu=n \log (2n)$ and
$D(2m) \leq 2m \log (4m)$, hence ${\hat H}_{\mathrm{rc}} = \BOM{m
  \lfloor (n \log n)/(p 2m \log (4m)) \rfloor}$.  If we let $m^*$
denote the value of $m$ for which the argument of the floor equals 1,
we can see that, for $m \leq m^*$, ${\hat H}_{\mathrm{rc}} = \BOM{(n
  \log n)/(p \log m)}$. In particular, for $m = \Theta(n/p)$, the
lower bound matches the upper bound~\eqref{eq:valiantub}.  For $m >
m^*$, the bound vanishes. At first, it may be puzzling that, for large
enough $m$, the memory based bound does not reproduce
bound~\eqref{eq:hk}; the reason is that, unlike the latter, the former
bound does not depend upon the assumption that, initially, only a
limited amount of input is available to each processor. 

A model with some similarities to BSP is the LPRAM
of Aggarwal et al.~\cite{AggarwalCS90}, where $p$ processors with unbounded local
storage are synchronized by a global clock and communicate via a
shared memory, rather than directly through a network. The
communication metric is the number of steps (cycles) $T_c$ in which at
least one processor reads from or writes to the shared memory.  A
straightforward adaptation of a well-known decomposition strategy for
the FFT achieves for $T_c$ an upper bound of the same form
as~\eqref{eq:valiantub}.  A lower bound of the same form is also
established. In addition to being developed for a different model, the
argument follows a route different than ours: a lower bound of the
same form is first established for sorting (assuming no input element
is ever kept by two processors at the same time), then claimed (by
analogy) for permutation networks, and finally adapted to the FFT
network, by exploiting the property established in~\cite{WuF81} that
the cascade of three FFT networks has the topology of a full
permutation network.

Finally, we mention that, motivated by the investigation of area-time
trade-offs in VLSI, a number of lower bounds have been established on
the information flow of the (multidimensional) discrete Fourier
transform (DFT) computed either exactly, on finite rings, or
approximately, on the complex field (see, e.g.,~\cite{BilardiF11} and
references therein).  These results apply to any algorithm, rather
than just to the radix-two specialization of the FFT considered in the
present paper and related work.  When adapted to computing the DFT on
BSP, with information measured in words capable of encoding a ring
element, the known $\BOM{n}$ word lower bound on the information flow
through the bisection of the system does imply that $H \geq {\hat
  H}_{\mathrm{rc}} = \BOM{n/p}$, assuming that no processor inputs (or
outputs) more than $n/2$ values. This bound is of the same order as
the dominator-based bound~\eqref{eq:hk} when $\log ((n/p) \log n ) =
\BOM{\log n}$ and is weaker otherwise, but it does hold under less
stringent constraints on the I/O protocol.
  
In summary, with respect to the outlined state of the art, our
objective is to develop lower bound techniques for the communication
complexity in BSP, $H$, capable to close the gap between the dominator
lower bound~\eqref{eq:hk} and the best known upper
bound~\eqref{eq:valiantub} for the FFT, as well as to weaken the
assumptions on the input protocol under which the lower bound is
established.

\subsection{Overview of Results}\label{sec:ourresults}

The main contribution of this paper is the \emph{switching potential
technique}, to obtain communication lower bounds for DAG
computations in the BSP model. The proposed technique applies to DAGs,
named switching DAGs, with $n$ input nodes where all nodes, except for
inputs and outputs, have out-degree equal to the in-degree.  Such a
graph $G=(V,E)$ can be viewed as a {switching network}~\cite{Savage98}
for which a switching size $N$ and a switching potential $\gamma$ are
defined.  The \emph{switching size} $N$ is the sum of the out-degrees
of the input nodes or, equivalently, the sum of the in-degree of the
output nodes.  The \emph {switching potential} $\gamma$ is the number
of different ways in which $N$ tokens initially placed on the $N$
outgoing arcs of the input nodes can be brought on the $N$ incoming
arcs of the output nodes, by moving them along arc-disjoint
paths. Intuitively, the switching potential is a measure of the
permuting ability of the graph. Its impact on the BSP communication complexity
is quantified in the following theorem.
\begin{restatable}{theorem}{mainlb}\label{thm:mainlb}
Let $\A$ be any algorithm that evaluates \emph{without recomputation}
a switching DAG $G=(V,E)$ on a BSP with $p$ processors. Let
$N$, $\gamma$, and $\Delta$ be respectively the switching size, the
switching potential, and the maximum out-degree of any node of $G$.
If the sum of the in-degree of the output nodes evaluated by every
processor is at most $U$, then the BSP communication complexity of
algorithm $\A$ satisfies
\begin{align}\label{eqn:mainlb}
H_{\A} \geq \left\{
 \begin{array}{ll}
 \frac{\log(\gamma/(U!)^{N/U})}{\Delta p \log (N/p)}& \,\text{if $p \leq N/e$,}\\
 \frac{e \log(\gamma/(U!)^{N/U})}{\Delta N \log e}& \,\text{otherwise.}
 \end{array}
 \right.
\end{align}
where, by definition, $\gamma \leq N!$ and $N/p \leq U \leq N$.
\end{restatable}
For the proof, we introduce the \emph{envelope game}, where a set of
envelopes, initially positioned on the input nodes, are moved to the
output nodes according to some given rules.  The evaluation of a DAG
in the BSP model is viewed as a run of this game; a lower bound is
derived for any BSP algorithm playing the game.

At the heart of the switching potential technique lies a counting
argument in combination with an indistinguishability argument,
somewhat similar to the approach used by Aggarwal and Vitter~\cite{AggarwalV88} to study
the I/O complexity of computations, and later applied to study the
complexity of communications in the LPRAM~\cite{AggarwalCS90}.
Our technique has the advantage that it can be directly
applied to any specific (switching) DAG, while the former approaches
require that a suitable combination of copies of the DAG under
consideration yield a full permutation network.

Even when $\gamma$ is large, bound~\eqref{eqn:mainlb} may become weak
due to high values of $U$. Intuitively, when $U$ is high, many
permutations can be realized by redistributing the data within the
processors, without much interprocessor communication. On the other
hand, if the different permutations realizable by a given DAG map
somewhat uniformly the inputs to the outputs, then the communication
complexity can be considerable, even for rather small values of
$\gamma$. This phenomenon has been investigated in~\cite{Vuillemin83}
for classes of permutations that form a transitive group. A special
case is the class of the cyclic shifts, which are permutations where
 all input values are shifted by a given amount (we refer to 
 Section~\ref{sec:cyclicshift} for a formal definition). 
 We establish the following result:
\begin{restatable}{theorem}{cycliclb}\label{thm:cycliclb}
Let $G=(V,E)$ be a DAG with $n$ input nodes and $n$ output nodes,
capable of realizing all the $n$ \emph{cyclic shifts}, with respect to
some fixed correspondence between inputs and outputs. Let $\A$ be an
algorithm that evaluates $G$ (possibly, \emph{with recomputation}) on
a BSP with $p \geq 2$ processors, such that initially each input is
available to exactly one processor.  Then, if some processor initially
stores (exactly) $q$ inputs, \emph{or} some processor evaluates
(exactly) $q$ outputs, the BSP communication complexity of algorithm
$\A$ satisfies
\[
H_{\A} \geq \frac{\min\{q,n-q\}}{2}.
\]
\end{restatable}
Some DAGs of interest happen to both exhibit a high switching
potential $\gamma$ and realize all cyclic shifts. For such DAGs, a
combination of Theorems~\ref{thm:mainlb} and~\ref{thm:cycliclb}
yields communication lower bounds under very mild assumptions
on the input/output protocol. An example is the $n$-input FFT, for
which $\gamma = 2^{n (\log n -1)}$ and the lower bound takes the
form, when $p \leq 2n/e$,
\begin{equation}\label{eqn:fftlb}
H_{\text{FFT}} > \frac{n \log(n/2)}{8 p \log(2 n/p)},
\end{equation}
as long as no processor evaluates more than $n/2$ output nodes, and
there is no recomputation. Lower bound~\eqref{eqn:fftlb} is the first
lower bound on the BSP communication complexity of the FFT that
asymptotically matches upper bound~\eqref{eq:valiantub} for any number
of processors $p \leq 2n/e$.  The technique based on the capability of
realizing all cyclic shifts also enables the extension of the
dominator-based lower bound~\eqref{eq:hk} to milder assumptions on the
input/output protocol. The dominator-based bound is not asymptotically
tight for $p = n/2^{o(\log \log n)}$, but it remains of interest when
recomputation is allowed.

We illustrate the versatility of the switching potential technique by
applying it to computations different from the FFT.  Sorting and
permutation networks naturally exhibit a high switching potential.
The corresponding BSP communication complexity lower bound is asymptotically
similar to bound~\eqref{eqn:fftlb}, and to the best of our knowledge
is the first known result for these computations. An asymptotically equivalent lower bound
was previously derived for BSP sorting in~\cite{Goodrich99}, but only for
algorithms with supersteps of degree $\BT{n/p}$ and input evenly
distributed among the processors. We also show how the switching
potential analysis can some time yield higher lower bounds if
separately applied to suitable parts of the DAG; in particular, we prove
for the BSP communication complexity of the periodic balanced sorting
network~\cite{DowdPRS89} a bound higher than the one derived for all
sorting networks.

The switching potential technique can be used, with some minor
changes, to derive lower bounds in other computational models besides
the BSP. Specifically, we apply the technique to a parallel variant of the I/O model
which includes, as special cases, both the I/O model and the LPRAM model.

In addition to the well-known general motivations for lower bound
techniques, we stress that striving for tight bounds for the whole
range of model's parameters has special interest in the study of
so-called \emph{oblivious} algorithms, which are specified without
reference to such parameters, but are designed with the goal of
achieving (near) optimality for wide ranges of the parameters.
Notable examples are cache-oblivious algorithms~\cite{FrigoLPR12},
multicore-oblivious algorithms~\cite{ChowdhuryRSB13},
resource-oblivious algorithms~\cite{ColeR12,ColeR17} and, closer
to the scenario of this paper, network-oblivious
algorithms~\cite{BilardiPPSS16}, where algorithms are
designed and analyzed on a BSP-like model. In fact, many BSP
algorithms are only defined or analyzed for a number of processors $p$
that is sufficiently small with respect to the input size $n$. For the
analysis of the FFT DAG, it is often assumed $p \leq \sqrt{n}$,
where the complexity is $\BT{n/p}$. Our results allow for the
removal of such restrictions.

A preliminary version of this paper appeared in~\cite{BilardiSS12}.
The current version contains an expanded discussion of previous
work, and provides full proofs of all claims, a significantly
simpler and slightly improved analysis of the switching potential technique,
applications of this technique to more case studies, as well as an adaptation
of it to a different model of computation.

\subsection{Paper Organization}
Section~\ref{sec:preliminary} introduces the BSP model and the concept
of switching DAG.  Section~\ref{sec:envelope} formulates the envelope
game, a convenient framework for studying the communication occurring
when evaluating a switching DAG.  Section~\ref{sec:switching} develops
the switching potential technique culminating with the proof of
Theorem~\ref{thm:mainlb}, which provides a lower bound on the BSP
communication complexity of a switching DAG, in terms of its switching
potential.  Section~\ref{sec:cyclicshift} establishes
Theorem~\ref{thm:cycliclb}, which provides a lower bound on the BSP
communication complexity of a DAG that can realize all cyclic shifts.
These two results are then applied, in
Section~\ref{sec:applications}, to the FFT DAG and to sorting
and permutation networks. Section~\ref{sec:io} extends 
the switching potential technique to computational models different from BSP.
Finally, in Section~\ref{sec:conclusions}, we draw some conclusions and discuss
directions for further work.

\section{Models of Computation}\label{sec:preliminary}
This section introduces the BSP model of parallel computation and the
class of computation DAGs for which our lower bound technique applies.

\subsection{The BSP Model}\label{sec:BSP}
The \emph{bulk-synchronous parallel} (BSP) model of computation was
introduced by Valiant~\cite{Valiant90} as a ``bridging model'' for
general-purpose parallel computing, providing an abstraction of both
parallel hardware and software.  It has been widely studied (see,
e.g.,~\cite{Tiskin11} and references therein) together with a number
of variants (such as D-BSP~\cite{delaTorreK96,BilardiPP07},
BSP*~\cite{BaumkerDH98}, E-BSP~\cite{JuurlinkW98}, and
BSPRAM~\cite{Tiskin98}) that aim at capturing data and communication
locality by basing the cost function on different communication
metrics.

The architectural component of the model consists of $p$ processors
$P_1,P_2,\dots,P_p$, each equipped with an unbounded local memory,
interconnected by a communication medium. The execution of a BSP
algorithm consists of a sequence of phases, called \emph{supersteps}:
in one superstep, each processor can perform operations on data
residing in its local memory, send/receive messages (each occupying a
constant number of words) to/from other processing elements and, at
the end, execute a global synchronization instruction.  Messages sent
during a superstep become visible to the receiver at the beginning of
the next superstep.

The running time of the $j$-th superstep is expressed in terms of two
parameters, $g$ and $\ell$, as $T_j = w_j + h_j g + \ell$, where $w_j$
is the maximum number of local operations performed by any processor
in the $j$-th superstep, and $h_j$ (usually called the \emph{degree}
of superstep $j$) is the maximum number of messages sent or received
by any processor in the $j$-th superstep. If the time unit is chosen to
be the duration of a local operation, then parameter $g$ is defined to be
such that the communication medium can deliver the messages of a
superstep of degree $h$ in $hg$ units of time, so that $1/g$ can be
viewed as measuring the available bandwidth of the communication medium,
whereas parameter $\ell$ is an upper bound on the time required for global
barrier synchronization. The \emph{running time} $T_{\mathcal A}$ of a
BSP algorithm $\mathcal A$ is the sum of the times of its supersteps and
can be expressed as $W_{\mathcal A} + H_{\mathcal A}g + S_{\mathcal
  A}\ell$, where $S_{\mathcal A}$ is the number of supersteps,
$W_{\mathcal A} = \sum_{j=1}^{S_{\mathcal A}} w_j$ is the \emph{local
  computation complexity}, and $H_{\mathcal A} =
\sum_{j=1}^{S_{\mathcal A}} h_j$ is the \emph{BSP communication
  complexity}. In this paper, we study the latter metric, which often
represents the dominant component of the running time.

\subsection{Switching DAGs}\label{sec:dag}
A \emph{computation DAG} $G = (V,E)$ is a directed acyclic graph where
nodes represent input and functional {operations} and arcs represent
{data dependencies}. More specifically, an arc $(u,v) \in E$ indicates
that the value produced by the operation associated with $u$ is one of
the operands of the operation associated with $v$, and we say that $u$
is a \emph{predecessor} of $v$ and $v$ a \emph{successor} of $u$. The
number of predecessors of a node $v$ is called its \emph{in-degree}
and denoted $\delta_{\mathrm{in}}(v)$, while the number of its
successors is called its \emph{out-degree} and denoted
$\delta_{\mathrm{out}}(v)$. A node $v$ is called an \emph{input} if
$\delta_{\mathrm{in}}(v)=0$ and an \emph{output} if
$\delta_{\mathrm{out}}(v)=0$.  We denote by $V_{\mathrm{in}}$ and
$V_{\mathrm{out}}$ the set of input and output nodes, respectively.
The remaining nodes are said to be \emph{internal} and their set is
denoted by $V_{\mathrm{int}}$.  

Of special interest for our developments are situations where the
computation executed by a given algorithm on a given input can be
viewed as embedding an \emph{evaluation} of a given computation DAG
$G$. Informally, this means that, during that execution, all the nodes
of $G$ are evaluated, respecting the dependencies specified by the
arcs. A bit more formally, we say that the execution of a BSP
algorithm on a specific input ${\bf x}$ \emph{evaluates} a given
computation DAG if, to any $v \in V$, is associated a set ${\cal
  S}(v)$ of processor-time pairs such that: (a) if $(t,P) \in {\cal
  S}(v)$, then, at time $t$, processor $P$ evaluates node $v$, by
either an input or a functional operation; (b) if $(t,P) \in {\cal
  S}(v)$ and $(u,v) \in E$, then there is a $(t',P') \in {\cal S}(u)$
such that the result of the evaluation of $u$ by $P'$ at time $t'$ is
effectively used as an operand by $P$ at time $t$.  Between time $t'$
and time $t$, the value of $u$ in question may be moved to different
processors and memory locations: the sequence of instructions
implementing such moves will be denoted as ${\cal
  S}(t',P',t,P)$. Taken together, the sets ${\cal S}$ define the
processing and communication \emph{schedule} of the evaluation of
$G$. We say that the evaluation is \emph{without recomputation} if
each node of $G$ is evaluated exactly once, that is, if ${\cal S}(v)$
is a singleton for every $v$.

A few observations may help provide some perspective on the
above notions of evaluation and schedule.  A little reflection will
show that an algorithm execution that embeds an evaluation of $G$
(possibly with recomputation) does not necessarily embed an
evaluation of $G$ without recomputation, whence forbidding
recomputation effectively restricts lower bounds results. We remark
that an execution of an algorithm that embeds an evaluation of $G$
may well contain additional operations not modeled by $G$ (for
example the additional operations may be instrumental to
constructing the DAG from the input, or the input size).  In
general, the execution of the same algorithm on different inputs may
embed the evaluation of different DAGs or of different schedules of
the same DAG. However, there are interesting algorithms that
evaluate the same DAG, with the same schedule, for all inputs of the
same size $n$.  Notable examples include the FFT, network algorithms
for sorting and permutations, standard matrix multiplication in a
semiring, and Strassen's matrix multiplication on a ring.

A number of graph-theoretic properties of the DAG can be related to
processing, storage, and communication requirements of the
underlying algorithm, as well as to its amount of parallelism. One
such property is the switching potential, which we introduce and
relate to the communication complexity of a BSP algorithm that
evaluates a DAG of the type defined next.

\begin{definition}\label{def:swDAG}
A \emph{switching DAG} $G = (V,E)$ is a computation DAG where, for any
internal node $v \in V_{\mathrm{int}}$, we have
$\delta_{\mathrm{out}}(v) =
\delta_{\mathrm{in}}(v)$.  We refer to $n=|V_{\mathrm{in}}|$ as to the
\emph{input size} of $G$, and introduce the \emph{switching size} $N$
of $G$ defined as
\[
N = \sum_{v \in V_{\mathrm{in}}} \delta_{\mathrm{out}}(v) = \sum_{v
  \in V_{\mathrm{out}}}
\delta_{\mathrm{in}}(v),
\]
where the equality between the two summations is easily established.
\end{definition}

Consider the set of arcs $E_{\mathrm{in}} = E \cap (V_{\mathrm{in}}
\times V)$, outgoing from the input nodes, and the set of arcs
$E_{\mathrm{out}} = E \cap (V \times V_{\mathrm{out}})$, incoming into
the output nodes. Let us now number both the arcs in $E_{\mathrm{in}}$
and those in $E_{\mathrm{out}}$ from $1$ to $N$, in some arbitrarily
chosen order. Then, to any partition of $E$ into a set of arc-disjoint
paths there corresponds a permutation $\rho =
(\rho(1),\rho(2),\ldots,\rho(N))$ of $(1,2,\dots,N)$, where $\rho(j)$
is the (number of the) last arc (in $E_{\mathrm{out}}$) of the unique
path whose first arc (in $E_{\mathrm{in}}$) is numbered $j$. See Figure~\ref{fig:rho}.

\begin{figure}
\centering
  \begin{subfigure}[b]{0.35\textwidth}
    \centering
    \resizebox{\linewidth}{!}{
      \begin{tikzpicture}[every node/.style={circle,draw},>=stealth,rotate=90]
		\node (a1) at (0,7) {};
		\node (a2) at (0,6) {};
		\node (a3) at (0,5) {};
		\node (a4) at (0,4) {};
		\node (b1) at (1.5,7) {};
		\node (b2) at (1.5,6) {};
		\node (b3) at (1.5,5) {};
		\node (b4) at (1.5,4) {};
		\node (c1) at (3.0,7) {};
		\node (c2) at (3.0,6) {};
		\node (c3) at (3.0,5) {};
		\node (c4) at (3.0,4) {};

		\draw[->] (a1) edge [densely dashed] (b1);
		\draw[->] (a2) edge [dashdotted] (b2);
		\draw[->] (a1) edge [loosely dashed] (b2);
		\draw[->] (a2) edge [dotted] (b1);
		\draw[->] (a3) edge (b3);
		\draw[->] (a4) edge [dashed] (b4);
		\draw[->] (a3) edge [densely dotted] (b4);
		\draw[->] (a4) edge [loosely dotted] (b3);
		\draw[->] (b1) edge [dotted] (c1);
		\draw[->] (b2) edge [loosely dashed] (c4);
		\draw[->] (b1) edge [densely dashed] (c3);
		\draw[->] (b2) edge [dashdotted] (c2);
		\draw[->] (b3) edge [loosely dotted] (c3);
		\draw[->] (b4) edge [dashed] (c4);
		\draw[->] (b3) edge (c1);
		\draw[->] (b4) edge [densely dotted] (c2);
      \end{tikzpicture}
    }
    \caption{$\rho = (5,7,1,3,2,4,6,8)$.}
    \label{a}
  \end{subfigure}
  \qquad\qquad
  \begin{subfigure}[b]{0.35\textwidth}
    \centering
    \resizebox{\linewidth}{!}{
      \begin{tikzpicture}[every node/.style={circle,draw},>=stealth,rotate=90]
		\node (a1) at (0,7) {};
		\node (a2) at (0,6) {};
		\node (a3) at (0,5) {};
		\node (a4) at (0,4) {};
		\node (b1) at (1.5,7) {};
		\node (b2) at (1.5,6) {};
		\node (b3) at (1.5,5) {};
		\node (b4) at (1.5,4) {};
		\node (c1) at (3.0,7) {};
		\node (c2) at (3.0,6) {};
		\node (c3) at (3.0,5) {};
		\node (c4) at (3.0,4) {};

		\draw[->] (a1) edge [densely dashed] (b1);
		\draw[->] (a2) edge [dashdotted] (b2);
		\draw[->] (a1) edge [loosely dashed] (b2);
		\draw[->] (a2) edge [dotted] (b1);
		\draw[->] (a3) edge (b3);
		\draw[->] (a4) edge [dashed] (b4);
		\draw[->] (a3) edge [densely dotted] (b4);
		\draw[->] (a4) edge [loosely dotted] (b3);
		\draw[->] (b1) edge [densely dashed] (c1);
		\draw[->] (b2) edge [loosely dashed] (c4);
		\draw[->] (b1) edge [dotted] (c3);
		\draw[->] (b2) edge [dashdotted] (c2);
		\draw[->] (b3) edge (c3);
		\draw[->] (b4) edge [densely dotted] (c4);
		\draw[->] (b3) edge [loosely dotted] (c1);
		\draw[->] (b4) edge [dashed] (c2);
      \end{tikzpicture}
    }
    \caption{$\rho = (1,7,5,3,6,8,2,4)$.}
    \label{b}
  \end{subfigure}
\caption{Examples of two different sets of arc-disjoint paths in the FFT DAG on $n = 4$ inputs and switching size $N = 8$.
Arcs in $E_{\mathrm{in}}$ and in $E_{\mathrm{out}}$ are numbered from $1$ to $N = 8$, left to right.}\label{fig:rho}
\end{figure}
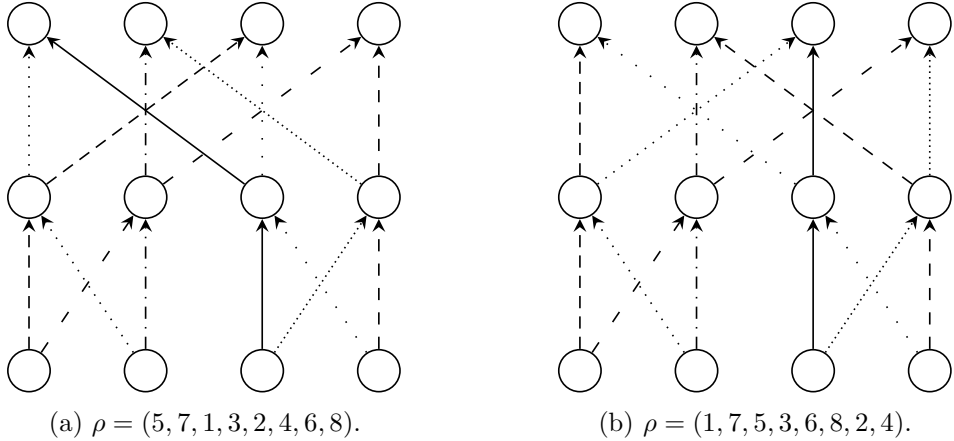

\begin{definition}\label{def:swPotential}
The \emph{switching potential} $\gamma$ of a switching DAG $G=(V,E)$
is the number of permutations $\rho$ corresponding to (one or more)
partitions of $E$ into arc-disjoint paths.
\end{definition}
Intuitively, if we think of each internal node $v$ of the DAG as a
switch that can be configured to connect its incoming arcs to its
outgoing arcs in any one-to-one correspondence, then switch
configurations uniquely correspond to partitions of $E$ into
arc-disjoint paths. Thus, $N$ items initially positioned on the input
nodes (specifically, $\delta_{\mathrm{out}}(u)$ items on input $u$)
can travel without conflicts and reach the output nodes. Indeed, in
the special case where $\delta_{\mathrm{out}}(u)=1$ for all input
nodes and $\delta_{\mathrm{in}}(v)=1$ for all output nodes, one has $N
= n = |V_{\mathrm{in}}| = |V_{\mathrm{out}}|$ and the switching DAG
becomes a switching network in the traditional sense~\cite{Savage98}.
When all permutations can be realized, that is $\gamma=n!$, then the
switching network is said to be a \emph{permutation} (or,
\emph{rearrangeable}) \emph{network}~\cite{Savage98}. It is a
simple exercise to establish that, for any switching DAG,
\[
\gamma \leq \prod_{v \in V \setminus V_{\mathrm{out}}} \delta_{\mathrm{out}}(v)
= \prod_{v \in V \setminus V_{\mathrm{in}}} \delta_{\mathrm{in}}(v),
\]
where each products equals the number of distinct partitions into
arc-disjoint paths. The inequality arises when distinct partitions
lead to the same permutation.

\section{The Envelope Game}\label{sec:envelope}

In this section we introduce the \emph{envelope game}, to be played on
a switching DAG.  The goal of the game consists in moving to the
output nodes some envelopes, initially placed on the input nodes,
according to some rules. Informally, the rules force the envelopes to
travel along arc-disjoint paths in a way that all the envelopes that
go through a given node must be simultaneously at that node, at some
time. The envelope game is meant to provide an abstraction of the
evaluation of a DAG without recomputation, which will prove useful in
the study of BSP communication complexity. The spirit is similar to
the one that motivated the introduction of the pebble game to study
space requirements of computations~\cite{PatersonH70,HopcroftPV77}.
\begin{definition}\label{def:envelopeGame}
The \emph{envelope game} on a switching DAG $G$ is defined by the
following rules, which characterize the (legal) \emph{runs} of
  the game. 
\begin{enumerate}
\item A set of $N$ distinguishable \emph{envelopes} is given, with
  exactly $\delta_{\mathrm{out}}(u)$ envelopes initially placed on
  each input node $u \in V_{\mathrm{in}}$ (hence, $N$ is the switching
  size of $G$).
\item The set of envelopes remains invariant during the game; at any
  stage each envelope is at exactly one node of $G$.
\item One move consists in moving one envelope from a node $u$ to a
  node $v$ along an arc $(u,v) \in E$.
\item An arc $(u,v)$ can be used only in one move.
\item An envelope can be moved from a node $u$ only after
  $\delta_{\mathrm{in}}(u)$ envelopes (each arriving from a different
  incoming arc of $u$) have been placed on $u$.
\item The (run of the) game is completed when all envelopes have
  reached an output node (i.e., when exactly $\delta_{\mathrm{in}}(w)$
  envelopes are placed on each output node $w \in V_{\mathrm{out}}$).
\end{enumerate}
\end{definition}
According to Rules 1, 2, 3, and 6, during one run, each envelope
traverses a path from $V_{\mathrm{in}}$ to $V_{\mathrm{out}}$.  Due to
Rule 4, paths traversed by different envelopes are arc-disjoint.  Due
to the property $\delta_{\mathrm{out}}(v)= \delta_{\mathrm{in}}(v)$ of
internal nodes of a switching DAG and to Rule 1 requiring that the
number of envelopes equals the switching size $N$, the paths traversed
by the envelopes yield a partition of the arc set $E$.  Therefore,
each run of the game uniquely identifies a permutation $\rho$
contributing to the switching potential, according to
Definition~\ref{def:swPotential}. Finally, we observe that Rule 5
requires that, for each node $u$, there is a time when all the
envelopes going through $u$ are on $u$.

We observe that, in spite of some similarities, the envelope
game differs from the pebbling game in various significant ways. In
particular, while the number of pebbles on the DAG can change during
the game and the goal is to minimize its (maximum) value, the number
of envelopes is constant throughout the game and the goal is to
count the mappings between the starting and the ending arcs of the
envelope paths by playing the game in all possible
ways. Furthermore, at any time, the number of pebbles on a DAG node
is either 0 or 1, whereas the number of envelopes on any given node
will typically go from 0 to the degree (by unit increments) and then
back to 0 (by unit decrements).

Intuitively, a run of the envelope game can be easily augmented into
an evaluation of the DAG.  We just need to imagine that each envelope
carries a (rewritable) card where, when a node $u$ is computed, its
result is written on the card. Since the envelopes leaving from a node
$u$ are distinct even though in the process of DAG evaluation they
would carry the same value, the communication of envelopes may result
in an overcounting of messages, by a factor at most $\Delta = \max_{v
  \in V} \delta_{\mathrm{out}}(v)$. In fact, there are at most
$\Delta$ envelopes with the same card value moving out from a
node. For many DAGs of interest, $\Delta$ is a small constant (e.g.,
for the FFT, $\Delta=2$), thus the overcounting is well bounded. At
the same time, the distinguishability of the envelopes simplifies the
analysis of the communication requirements of DAG evaluation. The next
lemma establishes that communication lower bounds can be transferred
from the envelope game to DAG evaluation, by showing how the former
can be obtained from the latter.

We say that a BSP algorithm \emph{plays} the envelope game on a
switching DAG $G$ if it satisfies the following conditions: (a) each
node $v \in V$ is assigned to a processor $P(v)$; (b) the envelopes
placed on input node $u$ are initially in $P(u)$; (c) whenever $(u,v)
\in E$ and $P(u) \not= P(v)$, the envelope moved along $(u,v)$ is sent
from $P(u)$ to $P(v)$ (possibly via intermediate processors), after
$P(u)$ has received all the envelopes destined to $u$.

\begin{lemma}\label{lem:game}
Let $\A$ be any algorithm that evaluates, without recomputation, a
switching DAG $G=(V,E)$ with $\Delta = \max_{v \in V}
\delta_{\mathrm{out}}(v)$, on a BSP with $p$ processors, and let
$H_{\mathcal A}$ be its BSP communication complexity. Then, there exists
an algorithm $\B$ that plays the envelope game on $G$ with BSP
communication complexity
\[
H_{\B} \leq \Delta H_{\A}.
\]
\end{lemma}

\begin{proof}
In algorithm $\B$, a node $v \in V$ is assigned to the processor
$P(v)$ where $v$ is evaluated by $\A$, which is unique due to the
hypothesis of no recomputation. (1) Initially, for every $u \in
V_{\mathrm{in}}$, $\delta_{\mathrm{out}}(u)$ envelopes are placed on
$u$ (as required by Rule 1) and each envelope is univocally assigned
to an outgoing arc of the respective input node. (2) The computation
in each internal node $u$ is replaced with a switch that sequentially
forwards the $\delta_{\mathrm{in}}(u)$ input envelopes to the
$\delta_{\mathrm{out}}(u)$ output arcs according to some
permutation. (3) For each arc $(u,v)$ where $P(v)$ differs from
$P(v)$, a message is sent from $P(u)$ to $P(v)$ (possibly via intermediate processors) carrying the envelope
moved along $(u,v)$.

We now show that algorithm $\B$ plays the envelope game. By
construction, the $N$ envelopes are set on the input nodes as required
by Rule 1. Rule 2 is satisfied since recomputation is disallowed
(i.e., no new envelope is added) and, in each internal node, all input
envelopes are forwarded to the outgoing arcs (i.e., no
envelope is deleted). Rules 3, 4, and 5 are complied with since a node
$u$ is computed in $\A$ only when all the $\delta_{\mathrm{in}}(u)$
inputs are ready (i.e., every predecessor of $u$ has sent an envelope
to
$u$ in $\B$) and the output values are propagated only to the
$\delta_{\mathrm{out}}(u)$ successors of $u$. Rule 6 is also obeyed
since all output nodes in $G$ are computed by $\A$.

We now show that the BSP communication complexity of $\B$ is at most
$\Delta$ times the BSP communication complexity of $\A$. The first two
modifications do not increase the communication (an envelope can be
locally constructed by a processor and no communication is required
for evaluating a switch).  The third modification can increases the BSP
communication complexity as analyzed next. Consider the case where a
node $u$ is processed by $\A$ on a processor, say $P_0$, while $\ell$
of its successors are processed on a different processor, say $P_1$,
where $1 \leq \ell \leq \Delta$. Then, one message from $P_0$ to $P_1$
is necessary and sufficient to send the output value of $u$ to the
$\ell$ successors in $P_1$ since the output value of a node is the
same for each successor. In contrast, since distinct envelopes are
sent by $u$ to its successors, $\ell$ messages must be sent by $\B$ in
order to forward the $\ell$ envelopes from node $u$ in $P_0$ to the
$\ell$ successors in $P_1$. Therefore, in the worst case, to each
message in $\A$ there correspond $\Delta$ messages in $\B$.
\end{proof}

Typically, a BSP algorithm $\A$ that evaluates a DAG will
implement the same schedule of operations and messages, for any
input of the DAG size $n$; in such a case, the corresponding
algorithm $\B$ that plays the envelope game will be the same for all
inputs as well.  However, Lemma~\ref{lem:game} applies even if DAG
$G$ is evaluated by $\A$ in different ways for different inputs
${\bf x}$, in which case the lemma inequality can be more explicitly
written as $H_{\B({\bf x})} \leq \Delta H_{\A({\bf x})}$.

\section{The Switching Potential Technique}\label{sec:switching}

This section develops the proof of Theorem~\ref{thm:mainlb}, a lower
bound on the communication complexity of any BSP algorithm $\A$ that
evaluates, without recomputation, a switching DAG $G$ with switching
potential $\gamma$. Technically, the lower bound individually
applies to any execution of the algorithm that embeds an evaluation
of $G$. We recall that $n$ denotes the input size and $N$ the
switching size of $G$.  Some observations may be useful to build up
intuition before entering the technical development, which
focuses on the envelope game.

One important point is that the communication requirements captured by
Theorem~\ref{thm:mainlb} do not simply arise from the data movement
implied by the $\gamma$ permutations that contribute to the switching
potential, but rather by the constraint that all those permutations
must be realizable under the same schedule, in the sense
defined in Section~\ref{sec:dag}. In fact, from one run of the
envelope game one can always obtain all other runs, by simply
changing the envelope permutation locally applied at each
node. Then, at any given time during the algorithm, the set of
locations that contain envelopes is the same for all runs, although
how the envelopes are distributed across those locations will differ
from run to run. Similarly, at any given superstep, the set of
(source, destination) pairs of the messages sent in that superstep
is the same for all runs, although the envelope carried by a given
message will generally vary from run to run.

To appreciate the implications of a fixed schedule, let us consider
the problem of permuting $N$ records among $p$ BSP processors, under
an input/output protocol where at most $q_{\mathrm{in}}$ records are
initially held by any processor and at most $q_{\mathrm{out}}$ records
are destined to any processor.  A simple BSP algorithm, executing just
one superstep where each record is directly sent from the source
processor to the destination processor, accomplishes the task with
minimum BSP communication complexity $\max\{q_{\mathrm{in}},q_{\mathrm{out}}\}$, assuming that a record fits
in one message.  In particular, if the I/O protocol is balanced, i.e.,
$q_{\mathrm{in}}=q_{\mathrm{out}}=N/p$, then each permutation can be accomplished with communication
complexity $N/p$.  However, under this algorithm, each permutation
will result in a different set of (source, destination) pairs for the
$N$ messages that carry the records, hence in a different
communication schedule.  On the other hand, a BSP algorithm that
evaluates the DAG corresponding to a sorting or permutation network
(see Section~\ref{sec:permnet} for more details) has a fixed
communication schedule, and the realization of a specific permutation
depends only on the content of the messages.

For the BSP model, we center our analysis around the
following quantity.
\begin{definition}
  Given an algorithm $\B$ that plays the envelope game on a switching
  DAG $G=(V,E)$, the \emph{redistribution potential} at superstep $j$,
  denoted $\eta_j$, is the number of different placements of the
  $N$ envelopes across the $p$ processors at the beginning of the
  $j$-th superstep that are achievable, in different runs, while
  complying with the schedule of $\B$. (The order of the envelopes
  within a processor is irrelevant.)
\end{definition}
  We let $\eta_{K+1}$ denote the number of different
  placements just after the end of the last superstep. The plan to
  establish Theorem~\ref{thm:mainlb} is along the following lines:
\begin{itemize}
\item First, we show that, without loss of generality, we can confine
  our analysis to algorithms where all supersteps have degree at most
  1. This greatly simplifies the subsequent counting arguments.
\item Second, we establish that $\eta_{K+1} \geq
  {\gamma}/{(U!)^{N/U}}$, for any algorithm $\B$ that plays the
  envelope game on DAG $G$ with switching potential $\gamma$, in terms
  of the maximum number $U$ of envelopes held by any processor at the
  end of the game.
\item Third, we establish that $\eta_{K+1} \leq (N/p)^{pH}$,
  due to the structure of the BSP model.
\item Finally, Theorem~\ref{thm:mainlb} stems from a combination of
  the upper and the lower bounds on $\eta_{K+1}$ in the two previous
  points.
\end{itemize}
The first three steps of the plan are each carried out by a separate
lemma. 

\begin{lemma}\label{lem:degree1}
For any BSP algorithm $\A$ with communication schedule independent of
the input, there exists a BSP algorithm $\A'$ with the following
properties.
\begin{itemize}
\item $\A'$ and $\A$ compute the same function.
\item The communication schedule of $\A'$ is independent of the input;
furthermore, in any superstep, each processor sends at most one
message and receives at most one message.
\item $\A'$ has the same BSP communication complexity as $\A$, that is,
$H_{\A'} = H_{\A}$.
\end{itemize}
\end{lemma}

\begin{proof}
We obtain $\A'$ from $\A$ by replacing each superstep ${\cal S}$ of
degree $h \geq 1$ with $h$ supersteps of degree $1$. (Supersteps where
$h=0$ are left unchanged.)  Specifically, let the communication
schedule of ${\cal S}$ be modeled by the bipartite \emph{message
  multigraph} $M = (S,D,F)$, where $S = \{s_1, \ldots, s_p\}$, $D =
\{d_1, \ldots, d_p\}$, and $F \subseteq S \times D$ where edge
multiset $F$ contains $(s_i,d_k)$ with a multiplicity equal to the
number of messages sent by $P_i$ to $P_k$ during ${\cal S}$. Clearly,
each node in $S$ and in $D$ has degree at most $h$. It is a simple
matter to augment $M$ with edges so as to obtain a regular bipartite
multigraph $M' = (S,D,F')$ of degree $h$.
We recall that Hall's theorem~\cite{Hall35} ensures that any regular
bipartite multigraph does admit a perfect matching.  Therefore, by
repeated applications of Hall's theorem, $F'$ can be partitioned into
$h$ perfect matchings, i.e., $F' = \sum_{\nu=1}^{h} F_{\nu}$, where
$M_{\nu} = (S,D,F_{\nu})$ is regular of degree $1$.  Correspondingly,
superstep ${\cal S}$ can be replaced by the equivalent sequence of
supersteps $({\cal S}_1, \ldots, {\cal S}_h)$ where, for
$\nu\in[h]$,\footnote{Throughout the paper,
  $[x]$ denotes the set $\{1,2,\dots,x\}$.}
  superstep ${\cal S}_{\nu}$ includes the messages of ${\cal
  S}$ corresponding to the edges in $F_{\nu} \cap F$.  (No message
corresponds to the edges in $(F'-F)$, which where added just to make
Hall's theorem directly applicable). All computations performed by
${\cal S}$ as well as the receive operations relative to messages sent
by the superstep preceding ${\cal S}$ are assigned to ${\cal S}_1$.
It is then straightforward to verify that the sequence $({\cal S}_1,
\ldots, {\cal S}_h)$ is equivalent to ${\cal S}$ and contributes $h$
to the BSP communication complexity of $\A'$, since each of its $h$
supersteps has degree $1$.
\end{proof}

\begin{lemma}\label{lem:etalb}
Consider any algorithm $\B$ for a BSP with $p$ processors that plays
the envelope game on a switching DAG $G=(V,E)$ in $K$ supersteps.
Let $N$ and $\gamma$ be the switching size and the
switching potential of $G$, respectively. At the end of the
algorithm, if each processor holds at most $U \leq N$ envelopes, the
redistribution potential satisfies
\begin{equation}\label{eqn:dplb}
\eta_{K+1} \geq \frac{\gamma}{(U!)^{N/U}}.
\end{equation}
\end{lemma}

\begin{proof}
Let $U_i \leq U$ denote the number of envelopes held by processor
$P_i$ just after the end of the last superstep of algorithm $\B$. We
know that, when varying the input-output correspondence for each
internal node of $G$ in all possible ways, $\gamma$ different
permutations of the envelopes over the arcs entering the outputs of
$G$ are generated.  At most $\Pi_{i=1}^{p}(U_i!)$ envelope
permutations can differ only by a rearrangement of the envelopes among
arcs assigned to the same processor and thus result in the same
placement of the envelopes across processors.  Hence,
\begin{equation}\label{eqn:dplb1}
\eta_{K+1} \geq \frac{\gamma}{\prod_{i=1}^{p}(U_i!)}.
\end{equation}
Considering that the quantity $\Pi_{i=1}^{p}(U_i!)$ is a superlinear
function of the $U_i$'s and that $\Sigma_{i=1}^{p}U_i = N$, a
convexity argument reveals that the maximum value is reached when
$\lfloor N/U \rfloor$ of the variables are set to the value $U$, one
variable is set to $(N \bmod U)$, and the remaining variables are set
to zero. Therefore,
\[
\prod_{i=1}^{p}(U_i!) \leq (U!)^{\lfloor N/U \rfloor} (N \bmod U)! \leq (U!)^{\lfloor N/U \rfloor} (U!)^{(N \bmod U)/U} = (U!)^{N/U},
\]
where the second inequality follows since the function $f(x) = \log
(x!)/x$ is increasing when $x>0$, and hence $f(N \bmod U) \leq
f(U)$. Then, by plugging inequality $\Pi_{i=1}^{p}(U_i!) \leq
(U!)^{N/U}$ in~\eqref{eqn:dplb1}, we obtain bound~\eqref{eqn:dplb}.
\end{proof}

\begin{lemma}\label{lem:etaub}
Consider any algorithm $\B$ for a BSP with $p$ processors that plays
the envelope game on a switching DAG $G=(V,E)$, with communication
complexity $H_{\B}$. Let $N$ be the switching size of $G$. At the end
of the algorithm, the redistribution potential satisfies
\begin{align*}\label{eqn:dpub}
\eta_{K+1} \leq \left\{
 \begin{array}{ll}
 (N/p)^{p H_{\B}}& \,\text{if $p \leq N/e$,}\\
 e^{N H_{\B}/e}& \,\text{otherwise.}
 \end{array}
 \right.
\end{align*}
\end{lemma}

\begin{proof}
For any $i \in [p]$ and $j \in [K+1]$, where $K$ is the number of
supersteps in $\B$, we denote with $t_{i,j}$ the number of envelopes
held by processor $P_i$ at the beginning of the $j$-th superstep.
Clearly, $\sum_{i=1}^{p} t_{i,j} = N$ for every $j$, since, by Rule 2
of the envelope game (Definition~\ref{def:envelopeGame}), the number
of envelopes is invariant and always equal to $N$.  (This constraint
would not necessarily hold if recomputation were allowed.)
Let, for every $j \in [K]$, $\mathcal{P}_j'$ to be the set of processors
holding at least one envelope at the beginning of the $j$-th superstep,
that is, $\mathcal{P}_j' = \{P_i : t_{i,j} \geq 1\}$. We claim that,
for every $j \in [K]$,
\begin{equation}\label{eqn:etas}
\eta_{j+1}/\eta_{j} \leq \prod_{i \in [p] : t_{i,j} \geq 1} t_{i,j}.
\end{equation}
Thanks to Lemma~\ref{lem:degree1}, we can assume without loss of
generality that all supersteps of $\B$ have degree at most one.
Consider superstep $j$, and consider a processor $P_i \in
\mathcal{P}_j'$.  At the beginning of the $j$-th superstep, $P_i$
holds $t_{i,j} \geq 1$ envelopes. From these, $P_i$ can choose, in
exactly $t_{i,j} \geq 1$ different ways, the envelope to send in the
$j$-th superstep, if any.  The claim then follows because any of the
$\eta_{j+1}$ envelope placements immediately after superstep $j$
correspond to one or more combinations of (a) one of the $\eta_j$
placements achievable immediately before superstep $j$, and (b) one
communication choice for each processor. Let $p_j' =
|\mathcal{P}_j'|$. Given that, for every $j \in [K]$, $\sum_{i \in [p]
  : t_{i,j} \geq 1} t_{i,j} = N$, the right-hand side
of~\eqref{eqn:etas} is maximized when all the $p_j'$ factors equal
$N/p_j'$, and hence from~\eqref{eqn:etas} we obtain
\[
\eta_{j+1}/\eta_{j} \leq (N/p_j')^{p_j'}. 
\]
Standard calculus reveals that the function $(N/p_j')^{p_j'}$ has
its maximum at $p_j' = N/e$. Therefore, since we must have
$p_j' \leq p$, we have
\begin{align*}
\eta_{j+1}/\eta_{j} \leq \left\{
 \begin{array}{ll}
 (N/p)^p& \,\text{if $p \leq N/e$,}\\
 e^{N/e}& \,\text{otherwise.}
 \end{array}
 \right.
\end{align*}
Multiplying both sides of the preceding relation over the $H_{\B}$
supersteps of degree one; considering that, for a superstep with
$h=0$, $\eta_{j+1}/\eta_{j} =1$; and observing that $\eta_1=1$, since
the only placement of envelopes among processors before the first
superstep is the one corresponding to the input placement protocol,
the claim follows.
\end{proof}

We are now ready to prove Theorem~\ref{thm:mainlb}, which we recall for convenience.

\mainlb*

\begin{proof}
When $p \leq N/e$, combining Lemma~\ref{lem:etalb} and Lemma~\ref{lem:etaub} yields
\[
(N/p)^{pH_{\B}} \geq  \frac{\gamma}{(U!)^{N/U}}.
\]
Taking the logarithm of both sides, solving for $H_{\B}$, and
recalling that, by Lemma~\ref{lem:game}, $H_{\A} \geq H_{\B}/\Delta$,
we conclude that
\[
H_{\A} \geq  \frac{\log (\gamma / (U!)^{N/U})}{\Delta p \log(N/p)}.
\]

The case $p > N/e$ is shown analogously.
\end{proof}

\section{The Cyclic Shift Technique}\label{sec:cyclicshift}

As discussed in Section~\ref{sec:ourresults}, the switching potential
lower bound becomes weaker as the maximum number $U$ of output nodes
held by a processor grows. In fact, the larger is $U$, the larger is
the number of permutations that can be realized without interprocessor
communication.  However, there are classes of permutations that, in
spite of their small cardinality, do require high communication even
for large values of $U$.  One such class is that of the
\emph{cyclic shifts} of order $n$, i.e., the permutations $\sigma_0,
\sigma_1, \ldots, \sigma_{n-1}$ such that, for $0 \leq k,i \leq n-1$,
we have $\sigma_k(i)=(i+k) \bmod n$. Intuitively, $\sigma_k$
cyclically shifts to the right, by $k$ positions. We say that a DAG
$G$ can realize all cyclic shifts of order $n$ if 1)
$|V_{\mathrm{in}}|=|V_{\mathrm{out}}|=n$; 2) there exist labelings
of the input nodes, $V_{\mathrm{in}} = \{v_0, v_1, \ldots,
v_{n-1}\}$, and of the output nodes, $V_{\mathrm{out}} = \{v'_0,
v'_1, \ldots, v'_{n-1}\}$, such that, for any $0 \leq k < n$, there
exists a set of $n$ arc-disjoint paths connecting $v_i\in
V_{\mathrm{in}}$ to $v'_{\sigma_k(i)}\in V_{\mathrm{out}}$, for any
$0 \leq i < n$. Several interesting computational DAGs do realize
all cyclic shifts of a given order. We now quantify the
communication required by algorithms that evaluate such DAGs.

We need to introduce the following notation. Let $P_1$ be one BSP
processor, and let $P_0$ be a virtual processor consisting of the other
$p-1$ processors. Denote by $I_1$ and $I_0$ the set of input nodes
initially held by $P_1$ and $P_0$, respectively. We also denote by
$O_1$ and $O_0$ the set of output nodes evaluated by $P_1$ and
$P_0$, respectively. 

\begin{lemma}\label{lem:cycliclb}
Let $G=(V,E)$ be a DAG with $n$ input nodes and $n$ output nodes,
capable of realizing all the $n$ \emph{cyclic shifts}, with respect to
some fixed labeling of inputs and outputs. Let $\A$ be an algorithm
that evaluates $G$ (possibly, \emph{with recomputation}) on a BSP with
$p \geq 2$ processors, such that initially each input is available to
exactly one processor.  Then the BSP communication complexity of $\A$
satisfies
\[
H_{\A} \geq \frac{|I_0||O_1|+|I_1||O_0|}{2n}.
\]
\end{lemma}

\begin{proof}
We use an argument patterned after Vuillemin~\cite{Vuillemin83}. There are $|O_1|$
cyclic shifts that match an input node in $I_0$ with every output node in
$O_1$; similarly, there are $|O_0|$ cyclic shifts that match an input node
in $I_1$ with every output node in $O_0$. Summing over all cyclic shifts,
we get that $|I_0||O_1|+|I_1||O_0|$ input nodes are assigned to output
nodes contained in a different processor. As there are $n$ cyclic shifts,
there must be a shift that matches $F =\lceil (|I_0||O_1|+|I_1||O_0|)/n\rceil$ input
nodes to output nodes contained in a different processor. Since at the
beginning of the computation each input is available to exactly one processor,
a total of $F$ messages are exchanged between $P_0$ and $P_1$. Therefore, $P_1$
receives or sends at least $F/2$ messages, and hence $H_{\A} \geq F/2$.
\end{proof}

As an immediate consequence of this lemma, we obtain the following.

\cycliclb*

\begin{proof}
Since, by hypothesis, each input is initially available to exactly one processor,
we have $|I_0|+|I_1| = n$. Moreover, since we can assume without loss of
generality that each output node is computed only once, $|O_0|+|O_1| = n$.
Thus, if we let $F$ denote the quantity $(|I_0||O_1|+|I_1||O_0|)/n$,
\[
F = \frac{(n-|I_1|)|O_1|+|I_1|(n-|O_1|)}{n} \geq \min\{|O_1|, n-|O_1|\}.
\]
Therefore, if $|O_1|=q$, then $F \geq \min\{q, n-q\}$. A symmetric
argument yields the same bound if $|I_1|=q$.  In conclusion, if
$|I_1|=q$ or $|O_1|=q$, by applying Lemma~\ref{lem:cycliclb}
we obtain $H_{\A} \geq F/2 \geq \min\{q, n-q\}/2$, as desired.
\end{proof}

\section{Applications}\label{sec:applications}
In this section, we show the versatility of the switching potential
technique by applying it to the FFT DAG and to sorting and permutation
networks.  
Further, we show that by applying the switching potential technique to
some parts of a particular sorting network it is possible to obtain a
lower bound stronger than the one obtained by applying the switching
potential technique to the entire DAG.

\subsection{Fast Fourier Transform}
Let $n$ be a power of two.
In the $n$-input FFT DAG, a node is a pair $\langle w,l \rangle$, with
$0 \leq w < n$ and $0 \leq l \leq \log n$, and there exists an arc
from node $\langle w,l \rangle$ to node $\langle w',l' \rangle$ if
and only if $l' = l+1$ and either $w$ and $w'$ are identical or their
binary representations differ exactly in the $l'$-th least significant bit.
See Figure~\ref{fig:FFT} for an example.

\begin{figure}[h]
\begin{center}
\resizebox{0.6\textwidth}{!}{
\begin{tikzpicture}[every node/.style={circle,draw},>=stealth,rotate=90]
\node (a1) at (0,7) {};
\node (a2) at (0,6) {};
\node (a3) at (0,5) {};
\node (a4) at (0,4) {};
\node (a5) at (0,3) {};
\node (a6) at (0,2) {};
\node (a7) at (0,1) {};
\node (a8) at (0,0) {};
\node (b1) at (1.5,7) {};
\node (b2) at (1.5,6) {};
\node (b3) at (1.5,5) {};
\node (b4) at (1.5,4) {};
\node (b5) at (1.5,3) {};
\node (b6) at (1.5,2) {};
\node (b7) at (1.5,1) {};
\node (b8) at (1.5,0) {};
\node (c1) at (3.0,7) {};
\node (c2) at (3.0,6) {};
\node (c3) at (3.0,5) {};
\node (c4) at (3.0,4) {};
\node (c5) at (3.0,3) {};
\node (c6) at (3.0,2) {};
\node (c7) at (3.0,1) {};
\node (c8) at (3.0,0) {};
\node (d1) at (4.5,7) {};
\node (d2) at (4.5,6) {};
\node (d3) at (4.5,5) {};
\node (d4) at (4.5,4) {};
\node (d5) at (4.5,3) {};
\node (d6) at (4.5,2) {};
\node (d7) at (4.5,1) {};
\node (d8) at (4.5,0) {};

\draw[->] (a1) edge (b1);
\draw[->] (a2) edge (b2);
\draw[->] (a1) edge (b2);
\draw[->] (a2) edge (b1);
\draw[->] (a3) edge (b3);
\draw[->] (a4) edge (b4);
\draw[->] (a3) edge (b4);
\draw[->] (a4) edge (b3);
\draw[->] (a5) edge (b5);
\draw[->] (a6) edge (b6);
\draw[->] (a5) edge (b6);
\draw[->] (a6) edge (b5);
\draw[->] (a7) edge (b7);
\draw[->] (a8) edge (b8);
\draw[->] (a7) edge (b8);
\draw[->] (a8) edge (b7);
\draw[->] (b1) edge (c1);
\draw[->] (b2) edge (c4);
\draw[->] (b1) edge (c3);
\draw[->] (b2) edge (c2);
\draw[->] (b3) edge (c3);
\draw[->] (b4) edge (c4);
\draw[->] (b3) edge (c1);
\draw[->] (b4) edge (c2);
\draw[->] (b5) edge (c5);
\draw[->] (b6) edge (c8);
\draw[->] (b5) edge (c7);
\draw[->] (b6) edge (c6);
\draw[->] (b7) edge (c7);
\draw[->] (b8) edge (c8);
\draw[->] (b7) edge (c5);
\draw[->] (b8) edge (c6);
\draw[->] (c1) edge (d1);
\draw[->] (c2) edge (d2);
\draw[->] (c1) edge (d5);
\draw[->] (c2) edge (d6);
\draw[->] (c3) edge (d3);
\draw[->] (c4) edge (d4);
\draw[->] (c3) edge (d7);
\draw[->] (c4) edge (d8);
\draw[->] (c5) edge (d5);
\draw[->] (c6) edge (d6);
\draw[->] (c5) edge (d1);
\draw[->] (c6) edge (d2);
\draw[->] (c7) edge (d7);
\draw[->] (c8) edge (d8);
\draw[->] (c7) edge (d3);
\draw[->] (c8) edge (d4);
\end{tikzpicture}}
\caption{The FFT DAG on $n=8$ inputs and switching size $N=16$. The
  nodes $\langle w,l \rangle$ are placed so that $w=0,\ldots,7$ from left
  to right and $l=0, \ldots, 3$ from bottom to top.}\label{fig:FFT}
\end{center}
\end{figure}
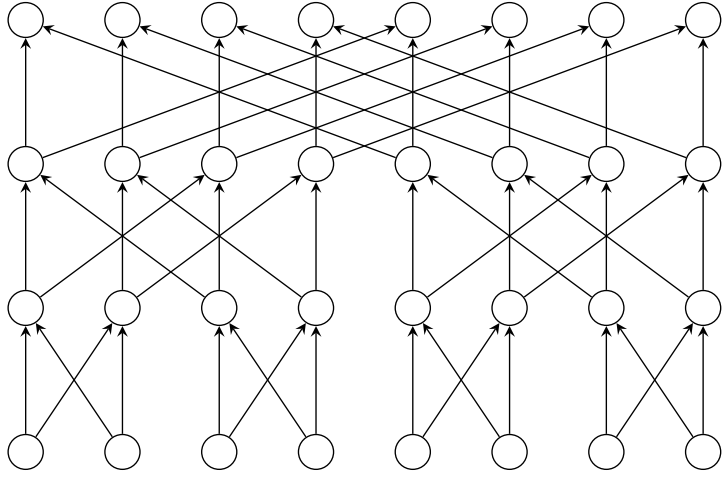

The $n$-input FFT DAG is a switching DAG since for any internal node $v$
we have $\delta_{\mathrm{in}}(v) = \delta_{\mathrm{out}}(v) = 2$,
and its switching size is $N = 2n$. Its switching potential is
established by the following lemma.

\begin{lemma}\label{lem:FFTsp} 
The FFT DAG of input size $n$ has switching potential $\gamma = 2^{n(\log n - 1)}$.
\end{lemma}

\begin{proof}
For each internal node, there exist two possible one-to-one relations
between the incoming arcs and the outgoing arcs.  A configuration of
the internal nodes is given by specifying the relation of each
internal node, and each configuration automatically defines a
particular set of $N = 2n$ arc-disjoint paths.  Since there are
$n(\log n-1)$ internal nodes in the FFT DAG, there are $2^{n(\log n -
  1)}$ possible configurations of the internal nodes.  No two
configurations define the same set: this follows as a corollary of the
property that in the FFT DAG there is a unique path between any input
node and any output node~\cite{Lawrie75}.\footnote{\cite{Lawrie75}
  discusses the property for the Omega network, which is isomorphic to
  the FFT network.}
\end{proof}

Now we show that the FFT DAG can realize all cyclic shifts. This will enable
the application of the results of Section~\ref{sec:cyclicshift}.

\begin{lemma}\label{lem:FFTshift}
The $n$-input FFT DAG can realize all cyclic shifts of order $n$.
\end{lemma}

\begin{proof}
In an $n$-input FFT DAG there exists a path from $\langle w,
0\rangle$ to $\langle (w+k) \bmod n, \log n \rangle$, for any $0 \leq
w < n$ and $0\leq k<n$.
This path visits $\langle f(w,l), l \rangle$  for each $0\leq l \leq
\log n$,
where $f(w,l)=\lfloor w/2^l\rfloor 2^l + ((w+k) \bmod 2^l)$. 
We now show that each one of the above sequences of nodes is
a connected path.
Clearly, we have $f(w,0)=w$ and $f(w,\log n)=(w+k) \bmod n$.
The values $f(w,l)$ and $f(w,l+1)$ differ at most in the $(l+1)$-th
least significant bit: indeed, the two values share the $l$ least
significant
bits since $f(w,l) \equiv f(w,l+1) \bmod 2^l$, and they also share the
$\log
n-l-1$ most significant bits since $\lfloor
f(w,l)/2^{l+1}\rfloor=\lfloor
f(w,l+1)/2^{l+1}\rfloor$. Therefore the DAG nodes $\langle f(w,l), l
\rangle$
and $\langle f(w,l+1), l+1 \rangle$ are connected and the path is well
defined.
Moreover, no two paths share a node for a given $k$: assume by
contradiction that there exist two paths $K_1=(\langle w,
0\rangle,\ldots,
\langle (w+k) \bmod n, \log n \rangle)$ and $K_2=(\langle w',
0\rangle,\ldots,
\langle (w'+k) \bmod n, \log n \rangle)$, with $w\neq w'$, that share
a node;
then, there must exist a value $l$, with $0< l < \log n$,  such that
$f(w,l)=f(w',l)$; however, since $w=f(w,l)-(k \bmod 2^l)$, it follows
that
$w=w'$ which is in contradiction with the initial assumption $w\neq
w'$; we can
thus conclude that paths $K_1$ and $K_2$ do not share any
node. Therefore, we
have that the FFT DAG can realize cyclic shifts for any
$0 \leq k < n$, by suitably setting DAG nodes according to the set of
$n$ paths specified by $k$.
\end{proof}

We are now ready to prove Theorem~\ref{thm:FFT}, which provides the
first lower bound on the BSP communication complexity required for
evaluating a $n$-input FFT DAG that asymptotically matches upper
bound~\eqref{eq:valiantub} for any number of processors $p \leq
2n/e$. Before this, we shall use the cyclic shift technique to
complement the dominator technique in order to obtain, for the FFT, a
lower bound of the same form of~\eqref{eq:hk} in
Section~\ref{sec:previous}, but under milder assumptions on the
input/output protocol. Since, like the dominator technique, the cyclic
shift technique does allow for recomputation, this results in a
strengthened lower bound for the general case when recomputation is
allowed, and thus is of independent interest.

\begin{theorem}
Let $\A$ be an algorithm that evaluates the $n$-input FFT DAG (possibly,
\emph{with recomputation}) on a BSP with $p \geq 2$ processors,
such that initially each input is available to exactly one processor.
Then, if no processor evaluates more than $n/\epsilon$ outputs, for
some constant $\epsilon > 1$, the BSP communication complexity of $\A$
satisfies
\[
H_{\A} = \BOM{\frac{n \log n}{p \log((n/p) \log n)}}.
\]
\end{theorem}

\begin{proof}
Since the $n$-input FFT DAG has $n (\log n - 1)$ internal nodes and $n$
output nodes, at least one processor has to evaluate $x = (n \log n)/p$ nodes.
Let $P_1$ be one of such processors, and let $P_0$ be a virtual processor
consisting of the remaining $p-1$ processors. Denote by $I_1$ and $I_0$
the set of input nodes initially held by $P_1$ and $P_0$, respectively.
We also denote by $O_1$ and $O_0$ the set of output nodes evaluated
by $P_1$ and $P_0$, respectively.

If $|I_1| \leq \beta x/\log x$, where $\beta$ is a suitably small constant,
then the lower bound in~\eqref{eq:hk} applies, and the theorem follows.
Otherwise, if $|I_1| > \beta x/\log x$, we shall leverage the hypothesis
whereby no processor evaluates more than $n/\epsilon$ outputs, for some
constant $\epsilon > 1$. This gives $|O_1| \leq n/\epsilon$, which in
turn implies $|O_0| \geq n - n/\epsilon = n (\epsilon - 1)/\epsilon$.
Then, since by hypothesis each input is initially available to exactly one
processor, and since by Lemma~\ref{lem:FFTshift} the FFT DAG can
realize all the $n$ cyclic shifts, we can apply Lemma~\ref{lem:cycliclb},
obtaining
\[
H_{\A} \geq \frac{|I_0||O_1|+|I_1||O_0|}{2n} \geq \frac{|I_1||O_0|}{2n} >
\frac{\beta n \log n  \cdot n (\epsilon - 1)}{2n \cdot p \log((n/p) \log n) \cdot \epsilon} =
\BOM{\frac{n \log n}{p \log((n/p) \log n)}},
\]
as desired.
\end{proof}

We now use the switching potential technique in synergy with
the cyclic shift technique to derive a tight lower bound
for the case when recomputation is disallowed. We shall consider
only the case $p \leq N/e$; nevertheless, Theorem~\ref{thm:mainlb}
also encompasses the case $p > N/e$. This case can be analyzed
in a similar way as we analyze the case $p \leq N/e$.

\begin{theorem}\label{thm:FFT}
Let ${\A}$ be any algorithm that evaluates without recomputation the
$n$-input FFT DAG on a BSP with $p \leq 2n/e$ processors, and let
$q$ be the maximum number of output nodes evaluated by a processor.
If initially each input is available to exactly one processor, then the communication
complexity of algorithm $\A$ satisfies
\[
H_{\A} \geq \frac{n \log(n/(8q^2))}{4 p \log(2 n/p)} + \frac{\min\{q, n-q\}}{4}.
\]
Moreover, if $q \leq n/2$,
\[
H_{\A} > \frac{n \log(n/2)}{{8} p \log(2 n/p)}.
\]
\end{theorem}

\begin{proof}
Since the in-degree of output nodes is two, we have that the sum of
the in-degree of the output nodes evaluated by each processor is at
most $U \leq 2q$. As recomputation is ruled out, we can apply
Theorem~\ref{thm:mainlb} with $N=2n$, $\Delta=2$, $U \leq 2q$, and
$\gamma = 2^{n(\log n -1)}$, obtaining, after some manipulations,
\[
H_{\A} \geq \frac{\log(\gamma/(U!)^{N/U})}{\Delta p \log (N/p)} \geq
\frac{n \log(n/(8 q^2))}{2 p \log(2 n/p)}.
\]
By hypothesis,  each input is initially available to exactly one processor,
and there exists some processor which evaluates (exactly) $q$ output nodes;
moreover, by Lemma~\ref{lem:FFTshift}, the FFT DAG can realize all cyclic shifts.
Therefore, we can apply Theorem~\ref{thm:cycliclb}, which gives
$H_{\A} \geq {(1/2)\min\{q, n-q\}}$.
By combining these two lower bounds we obtain the first claim of the theorem.

Consider now the case $q \leq n/2$. In this case we can write
\[
H_{\A} \geq \frac{n \log(n/(8q^2))}{4 p \log(2 n/p)} + \frac{q}{4} >
\frac{n \log(n/(8 q^2))}{8 p \log(2 n/p)} + \frac{q}{4}.
\]
Let us denote with $H'(q)$ the rightmost term of the above
inequality, that is,
\[
H'(q) =  \frac{n \log(n/(8q^2))}{8 p \log(2 n/p)} + \frac{q}{4}.
\]
By deriving $H'(q)$ with respect to $q$, we can see that
$H'(q)$ is non-decreasing for $q \geq n/(p\log(2n/p))$.
Therefore, since $q \geq n/p$, and since $n/p > n/(p\log(2n/p))$
(because of the $p \leq 2n/e$ hypothesis), we have that
$H'(q) \geq H'(n/p)$, and thus
\[
H_{\A} > H'(q) \geq H'(n/p) = \frac{n \log(n/2)}{8 p \log(2 n/p)},
\]
which proves the second claim of the theorem.
\end{proof}

\subsection{Sorting and Permutation Networks\label{sec:permnet}}

In this section we apply our technique to bound from below the
BSP communication complexity of the computation DAGs that
correspond to sorting and permutation networks. These networks,
such as the Bene\v{s} permutation network~\cite{Benes64} and the
bitonic~\cite{Batcher68} and AKS sorting networks~\cite{AjtaiKS83},
can be interpreted as switching DAGs and have the property
that they can realize all the possible permutations, and thus
the switching potential technique can be naturally applied to them. 
Since our technique abstracts the DAG under consideration by
considering only one general parameter (its switching potential),
the lower bound obtained in this section is \emph{universal} in
the sense that it holds for \emph{any} sorting or permutation network.

We now briefly recall the definitions of such networks. More complete
descriptions can be found in~\cite{Knuth73,Leighton92,Savage98}.
A \emph{comparator network} is an acyclic circuit of
\emph{comparators}. A comparator is a $2$-input
$2$-output operator which returns the minimum of the two inputs on one
output, and the maximum on the other.
An $n$-input comparator network is called a \emph{sorting network} if
it produces the same output sequence
on all $n!$ permutations of the inputs. Thus, sorting networks can be
seen as a simple model for data-oblivious
sorting algorithms, that is, algorithms that perform the same set of
operations for all values of the input data.
A \emph{routing network} is an acyclic circuit of \emph{switches}. A
switch is a $2$-input $2$-output operator
which either passes its two inputs to its outputs, or it swaps
them. An $n$-input routing network is called a
\emph{permutation network} if for each of the $n!$ permutations of the
inputs there exists a setting of the
switches that creates $n$ disjoint paths from the $n$ inputs to the
$n$ outputs. Observe that, suitably modified
to transmit messages, every sorting network is a permutation network,
but the converse is not true.

A sorting/permutation network can be naturally modeled as a
computation DAG by associating to each comparator/switch a pair of
(internal) nodes of the DAG. Both nodes have the same two
predecessors, and thus receive in input the same two values $a$ and
$b$, and both nodes have out-degree two, but one node computes the
function $x = \min\{a,b\}$ while the other computes the function
$y = \max\{a,b\}$. The resulting DAG is therefore a switching
DAG of input size $n$, switching size $N = 2n$, and with $\Delta = 2$.
The following lemma shows that the DAG of any sorting or permutation
network has switching potential $\gamma \geq n!$, and can realize all
cyclic shifts.

\begin{lemma}
The DAG of any sorting or permutation network with input size $n$ has
switching potential $\gamma\geq n!$ and can realize all cyclic shifts
of order $n$.
\end{lemma}

\begin{proof}
Since any sorting or permutation network can perform all the $n!$
permutations of $n$ inputs, there exist $n!$ sets $\mathcal S$
of $n$ arc-disjoint paths connecting input and output nodes. 
Each set $S \in \mathcal S$ of $n$ arc-disjoint paths determines a set
of $n$ paths from $n$ outgoing arcs of the input nodes to $n$ incoming
arcs of the output nodes. 
In order to get the claimed switching potential, we have to construct
additional $n$ paths from the $n$ outgoing arcs of the input nodes to
the $n$ incoming arcs of the output nodes that are not used in $S$.
We observe that each path in a $S\in \mathcal S$ uses only one of the
two incoming arcs of each (non-input) node, and only one of the two
outgoing arcs of each (non-output) node.
By exploiting the unused pair of incoming/outgoing arcs in each
internal node, it is possible to uniquely construct the missing paths.
Therefore, there exists $n!$ sets of $N=2n$ arc-disjoint paths
connecting input and output nodes, and the first part of the claim
follows.\footnote{
In a DAG corresponding to a sorting or permutation network with $n$ inputs 
and internal nodes of in-degree $\delta$, the switching size is $N=\delta n$ 
since there are $\delta$ outgoing edges per input node.
Therefore, the switching potential $\gamma$ can be as large as $N!=(\delta n)!$.
The present argument shows that $\gamma$ is at least $n!$. 
}
Furthermore,  the set $\mathcal S$ contains the set of $n$
arc-disjoint paths of all the $n$ cyclic shifts by definition of
sorting and permutation networks, and the second part of the claim
follows as well.
\end{proof}

We are therefore in a similar situation as for the FFT DAG, with a
slightly different switching potential $\gamma$,
and it is therefore sufficient to mimic the proof for
Theorem~\ref{thm:FFT}. We have the following result.

\begin{theorem}\label{thm:networks}
Let $\mathcal{N}$ be any sorting or permutation network with $n$ inputs.
Let $\A$ be any algorithm that evaluates without recomputation the DAG
corresponding to $\mathcal{N}$ on a BSP with $p \leq 2n/e$ processors,
and let $q$ be the maximum number of output nodes
evaluated by a processor. Then the BSP communication complexity
of algorithm $\A$ satisfies
\[
H_{\A} \geq \frac{n \log(n/(4 e q^2))}{4 p \log(2 n/p)} +
\frac{\min\{q, n-q\}}{4}.
\]
Moreover, if $q \leq n/2$,
\[
H_{\A} > \frac{n \log(n/e)}{8 p \log(2 n/p)}.
\]
\end{theorem}

\begin{proof}
Analogous to the proof of Theorem~\ref{thm:FFT}, with $\gamma \geq n!$
in place of $\gamma=2^{n(\log n -1)}$.
\end{proof}

We observe that for sorting and permutation networks we are using a
lower bound on the switching potential
$\gamma$ which is lower (when $n \geq 6$) than the value of the
switching potential of the FFT DAG, and
then the resulting lower bound on the BSP communication complexity has a
lower constant inside the logarithmic
term at the numerator with respect to the one in
Theorem~\ref{thm:FFT}. This is due to the generality of our
argument that applies to the entire family of sorting and permutation
networks. Better bounds can be obtained
for specific networks. For example, in the case of the Bene\v{s}
permutation network, the same bound of
Theorem~\ref{thm:FFT} applies since the corresponding DAG contains an
$n$-input FFT DAG.

Finally, we observe that the lower bound of Theorem~\ref{thm:networks}
is asymptotically tight, as the
computation DAG corresponding to the Bene\v{s} permutation network can
be evaluated with the same
strategy used for the FFT DAG, yielding a BSP communication complexity of
$O(n \log n/(p \log (n/p)))$.

\subsection{Boosting the Switching Potential Technique}
Since $\gamma \leq N!$, if applied to an entire switching DAG,
Theorem~\ref{thm:mainlb} cannot yield a lower bound larger than
$\BOM{(N \log N)/(p \log(N/p))}$.  However, by applying the theorem to
suitable parts of the DAG and composing the results, it is sometime
possible to obtain asymptotically larger lower bounds.  To illustrate
the approach, we study the DAG of the \emph{periodic balanced sorting
network} (PBSN)~\cite{DowdPRS89}, which consists of a sequence of
$\log n$ identical \emph{blocks}. Specifically, we consider the case
where $n$ is a power of two and the block is the \emph{balanced
merging network} (BMN), as in~\cite{DowdPRS89}.  (The analysis and
the result would also apply when the block is the odd-even merging
network).

The DAG of an $n$-input BMN is the following: a node is a pair
$\langle w,l \rangle$, with $0 \leq w < n$ and $0 \leq l \leq \log n$;
there exists an arc between two nodes $\langle w,l \rangle$ and
$\langle w',l' \rangle$ if and only if $l' = l+1$, and either $w=w'$
are identical or $w'=(i+1) n/2^l - j$ where $i=\lfloor w2^l/n\rfloor$
and $j= i \bmod (n/2^l)$.  The $n$-input BMN is a switching DAG since
for any internal node $v$ we have $\delta_{\mathrm{in}}(v) =
\delta_{\mathrm{out}}(v) = 2$, and its switching size is $N = 2n$.

The DAGs of the BMN and of the FFT are isomorphic, that is, there
exists an arc-preserving bijection between the two node
sets~\cite{Bilardi89}.  Hence, the BMN has the same switching
potential of the FFT DAG (i.e., $\gamma=2^{n (\log n-1)}$, see
Lemma~\ref{lem:FFTsp}) and can realize all cyclic shifts of order $n$
(see Lemma~\ref{lem:FFTshift}).  As a consequence, the lower bounds
stated in Theorem~\ref{thm:FFT} for the FFT DAG apply unchanged to the
DAG of the BMN.

By separately applying the switching potential technique to each of
the $\log n$ BMN blocks of a PBSN, we obtain the following result.

\begin{theorem}\label{thm:BSN}
Let ${\A}$ be any algorithm that evaluates without recomputation
the DAG of the $n$-input PBSN, where the block is a BMN,
on a BSP with $p \leq 2n/e$ processors. Let $q$ be the
maximum number of output nodes of each block evaluated by a processor.
If each input is initially available to exactly one processor and $q \leq n/2$,
then the BSP communication complexity of $\A$ satisfies
\[
H_{\A} > \frac{n \log(n/2)}{8 p \log(2 n/p)}\left \lceil \frac{\log n}{2}\right \rceil.
\]
\end{theorem}

\begin{proof}
Consider the sequence $s_1,s_2,\dots,s_{\log n}$ of BMNs in the PBSN.
For any $1 \leq i < \log n - 1$, the evaluations of $s_{i}$ and
$s_{i+2}$ in algorithm $\A$ cannot overlap in time: as BMN sorts any
bitonic sequence~\cite{DowdPRS89}, any input value can reach any
output; therefore, no input value of $s_{i+2}$ is ready until all
output values of $s_{i}$ have been computed.  It follows that the
evaluations of the $\lceil \log n / 2 \rceil$ odd BMNs cannot overlap
in time. Since, by hypothesis, no processor evaluates more than $q
\leq n/2$ output nodes of each BMN, we can apply (the second part of)
Theorem~\ref{thm:FFT} to each BMN, and the claim follows.
\end{proof}

Observe that the PBSN contains $\Theta(n \log^2 n)$ comparators,
which is a factor $\log n$ more than the optimal value. Therefore,
it is natural that the BSP communication complexity for evaluating a
PBSN is a factor $\Omega(\log n)$ larger than the lower bound that
holds for any sorting network (Theorem~\ref{thm:networks}).
Nevertheless, a lower bound of the form of the one given in
Theorem~\ref{thm:BSN} cannot be derived by applying the switching
potential technique to the entire DAG of a PBSN. We are not aware of
any prior lower bounds of this form for computations that correspond to
the evaluation of PBSN.

\section{The Switching Potential Technique in a Parallel I/O Model}\label{sec:io}

In this section, we show how the switching potential technique
can be adapted to yield lower bounds in computational models
different from BSP.  Specifically, we consider a parallel variant of the
\emph{I/O model}, which includes, as special cases, both the I/O
model of Hong and Kung~\cite{HongK81} and the LPRAM model of Aggarwal et al.~\cite{AggarwalCS90}.

Our parallel I/O model consists of $p$ synchronous processors, each
with a (fast) local memory of $m$ words, which can access a (slow)
shared memory of (potentially) unbounded size. In each step, all the
processors perform the same instruction, which can be (i) an operation
on data in the local memory, (ii) a move of a word from the shared to
the local memory (i.e., a \textit{read} operation) or (iii) vice versa
(i.e., a \textit{write} operation).\footnote{Our lower bound can be
  trivially adjusted, dividing it by $b$, if an instruction can move
  $b$ memory words instead of just one.}  The \emph{I/O
  complexity} $H_\A$ of an algorithm $\A$ is the number of steps where
a read or write operation occurs.\footnote{We use here the same
  notation $H_\A$ for the I/O complexity as for the BSP communication
  complexity, to highlight the similar role of the two metrics in the
  context of the switching potential technique.}  We assume the input
and output of algorithm $\A$ to reside in the shared memory at the
beginning and at the end of $\A$, respectively; (if this is not the
case, our lower bound may still apply, after suitable modifications).
We refer to~\cite{Vitter06} for a survey on algorithms and data
structures for the I/O model.

Consider now an algorithm $\B$, for the parallel I/O model,
playing the envelope game on a switching DAG $G=(V,E)$, with
switching potential $\gamma$.  We assume that each envelope occupies
one memory word, whence there cannot be more than $m$ envelopes in
any local memory.  Due to the restrictions of the I/O model, algorithm
$\B$ can move the first envelope from a node $u$ only when all the
$\delta_{\mathrm{in}}(u)$ input envelopes are in the local memory of the same
processor.

As for the BSP model, a crucial observation is that the switching
potential technique does not simply arise from the data movement
implied by the $\gamma$ permutations that contribute to the switching
potential, but rather by the constraint that all those permutations
must be realizable under the same schedule. Formally, the
\emph{schedule} of an algorithm is defined by the sequence of read and
write operations, and by the memory locations in the shared and local
memories involved in each read and write operation.  Since the
schedule is given, only the content (not the source and the
destination) of a read and write operation can vary across different
runs of the envelope game that can result in the realization of
different permutations. Thus, at any given time, the memory locations
(both shared and local) that contain envelopes are independent of the
run of the game, while the mapping of the envelopes to those location
will generally differ across runs.  This fact allows us to introduce a
notion of redistribution potential, appropriate for the parallel I/O
model:
\begin{definition}
Consider an algorithm $\B$ for the parallel I/O model that plays the
envelope game on a switching DAG $G=(V,E)$. At any given step of the
algorithm, the \emph{envelope placement} is the specification, for
each envelope, of either the (address of the) shared memory location
or the (index of the) processor whose local memory contains that
envelope. (The exact position of the envelope within a local memory is
irrelevant.) The \emph{redistribution potential} at the $j$-th
read/write operation, denoted $\eta_j$, is the number of different
envelope placements, before the $j$-th read/write operation, that
are achievable in different runs, while complying with the schedule
of $\B$.
\end{definition}
From the initial and final conditions of the game, we have $\eta_1=1$
and $\eta_{H_\B+1}\geq \gamma$.

Intuitively, read and local operations do not increase the
redistribution potential. On the other hand, the increase due to write
operations is bounded by the amount of envelopes held locally by each
processor. This amount is naturally bounded by the size $m$ of the
local memory, but it is also effectively bounded by $N/p$, since the
$p$ processors together can at most hold $N$ envelopes, and it turns
out that a balanced allocation of envelopes to processors can result
in the maximum increase of redistribution potential.  These statements
are substantiated in the next lemma, leading to a lower bound on the
number of write steps needed to bring the redistribution potential
from $1$ to $\gamma$, hence to the I/O complexity.

\begin{lemma}\label{lem:io_pot}
The I/O complexity of an algorithm $\B$ that plays the envelope game
on a switching DAG $G = (V,E)$ in the parallel I/O model with $p$
processors, with local memory size $m$, satisfies
\[
H_\B \geq \frac{\log \gamma}{p \log \min\{m, N/p\}},
\]
where $\gamma$ is the switching potential of $G$.
\end{lemma}

\begin{proof}
An operation on data in the local memories cannot change the
redistribution potential, since it does not affect the shared memory
and does not change the set of envelopes present in each local memory.
A read operation moves memory words from the shared memory to the
local memories, in a way that is uniquely prescribed by the schedule
of the algorithm.  Hence, to each envelope placement before the
operation there correspond a unique placement after the operation,
whence the redistribution potential does not increase.
  
On the other hand, each write operation increases the redistribution
potential by at most a factor $\prod_{i=1}^{p}N_i$, where $N_i$
denotes the number of envelopes currently in the $i$-th processor.  In
fact, the envelope to be moved to shared memory by the $i$-th
processor can only be one of the $N_i$ envelopes in its local memory.
(For clarity, we observe that the source and the target addresses of
the write are actually fixed by the schedule.  It is the envelope
currently located at the source address that can differ across
different runs of the game, due to permutations of envelopes within a
local memory.)

We now observe that, on the one hand, for each $i$, we have
$N_i \leq m$, whence $\prod_{i=1}^{p}N_i \leq m^p$. On the other
hand, by the invariance of the number of envelopes, we have
$\sum_{i=1}^{p}N_i =N$, a constraint under which $\prod_{i=1}^{p}N_i
\leq (N/p)^p$, as can be established by standard techniques.
We then have $\eta_{j+1} \leq (\min\{m,N/p\})^p \eta_j$, whence
$\eta_{H_\B+1} \leq \min\{m,N/p\}^{pH_\B}$. Since $\eta_{H_\B+1} \geq
\gamma$, the claim follows.
\end{proof}

We are now ready to provide a lower bound on the I/O
complexity of any algorithm evaluating a switching DAG $G$ without recomputation.

\begin{theorem}\label{thm:parIO}
Let $\A$ be any algorithm that evaluates \emph{without recomputation} a switching DAG 
$G = (V,E)$ in the parallel I/O model with $p$ processors and with local memory size $m$.
Let $N$, $\gamma$, and $\Delta$ be the switching size of $G$, the switching potential of $G$,
and the maximum out-degree of any node of $G$, respectively. 
Then the I/O complexity of algorithm $\A$ satisfies
\[
H_\A \geq \frac{\log \gamma}{\Delta p \log \min\{\Delta m, N/p\}}.
\]
\end{theorem}

\begin{proof}
An algorithm $\A$ that evaluates $G$ on the parallel I/O model with local memory size $m$ and $p$ processors can be 
transformed into an algorithm $\B$ that plays the envelope game on the parallel I/O model with local memory size $\Delta m$ and $p$ processors.
Algorithm $\B$ is obtained from $\A$ with the following three changes. (1) Initially, for every $u \in
V_{\mathrm{in}}$, $\delta_{\mathrm{out}}(u)$ envelopes are placed on
$u$ and each envelope is univocally assigned
to an outgoing arc of the respective input node. (2) The computation
in each internal node $u$ is replaced with a switch that sequentially
forwards the $\delta_{\mathrm{in}}(u)$ input envelopes to the
$\delta_{\mathrm{out}}(u)$ output arcs according to some
permutation. This replacement requires that $\delta_{\mathrm{out}}(u) \leq \Delta$
 words are available to store the envelopes forwarded by $u$, whereas 
 only one word was required to  store the single output of node $u$: since the local memory is $\Delta$ times larger than the one used by algorithm $\A$, there is enough space to store the at most
$\Delta$ envelopes for each of the $m$ nodes kept in the local memory by $\A$.
 (3) For each arc $(u,v)$ where $P(v)$ differs from
$P(v)$, the envelope is first written in the shared memory by $P(u)$ and then read by $P(v)$ (possibly, other processors can read and write the envelope between these two operations).
As shown in Lemma~\ref{lem:game}, the above modifications guarantee that the six rules of the envelope game are satisfied.

We now observe that $H_\B\leq \Delta H_\A$. The first two changes do not increase the I/O complexity. On the other hand, the third change increases the I/O complexity by a factor $\Delta$: indeed, for each node $u$, the output values on the  $\delta_{\mathrm{out}}(u)$ outgoing edges can be stored by algorithm $\A$ in one word of the local/shared memory since the output values are indistinguishable; however, algorithm $\B$ requires $\delta_{\mathrm{out}}(u)\leq \Delta$ words since envelopes are distinct; therefore the I/O complexity of $\B$ is at most $\Delta H_\A$.

The theorem follows since, by Lemma~\ref{lem:io_pot}, an algorithm $\B$ playing the envelope 
game on an $\Delta m$-word local memory with $p$ processors requires 
$H_\B \geq (\log \gamma)/{(p \log \min\{\Delta m, N/p\})}$ I/Os.
Since $H_\B\leq \Delta H_\A$, the main claim follows.
\end{proof}

We now analyze the above lower bounds for two special cases: the sequential I/O model~\cite{HongK81} and the LPRAM~\cite{AggarwalCS90}.
The I/O model follows by the parallel I/O model by setting $p=1$. 
If we consider the FFT DAG and assume that $N \geq 2m$, we obtain a lower 
bound of $n \log (n/2) / (2\log (2m))$, which asymptotically matches the lower bound in~\cite{HongK81}.
On the other hand, the LPRAM is obtained by setting $m=+\infty$. In this case, we get for the FFT DAG a lower bound of $n \log (n/2) / (2p \log (2n/p))$, which asymptotically matches the lower bound in~\cite{AggarwalCS90}.

Interestingly, we observe that Theorem~\ref{thm:parIO} gives a tight
lower bound for the FFT DAG without additionally resorting to the
cyclic shift technique, as we did in the BSP model.  The reason for
such behavior is that the I/O protocol requires envelopes to be stored
in the shared memory at the beginning and end of the algorithm. Thus,
the redistribution potential cannot be increased without I/O
operations.  In contrast, in the BSP model, envelopes are contained in
the processors' local memories since there is no ``external storage''
where to store envelopes at the beginning and at the end of the
algorithm.  Therefore, if each BSP processor contains at most $U$ of
the $N$ envelopes, it is possible to get $(U!)^{N/U}$ permutations
even without communication, by just rearranging envelopes within each
processor (see the proof of Lemma~\ref{lem:etalb}). If $U$ is
sufficiently large compared to the switching potential $\gamma$,
almost all permutations can be reached without communication and we
then need the cyclic shift technique to reinforce the lower bound.

\section{Conclusions}\label{sec:conclusions}
In this paper we have studied some aspects of the complexity of
communication of parallel algorithms. We have presented new
techniques for deriving lower bounds on communication complexity
for computations that can be represented by a certain class of DAGs.
We have demonstrated the effectiveness of this technique by
deriving novel, mostly tight lower bounds for the FFT and for
sorting and permutation networks.

The present work can be naturally extended in several directions, some
of which are briefly outlined next. First, it would be interesting to
apply the switching potential technique to other DAG computations
beyond the few case studies of this paper. One example are DAGs that
correspond to merging networks. We conjecture that the switching
potential $\gamma$ of any DAG corresponding to a merging network of
input size $n = 2^k$ satisfies $\log \gamma = \BOM{n \log n}$; we also
conjecture that the same bound holds for any network that can realize
all cyclic shifts.  It is also natural to explore the application of
the switching potential technique to other models for distributed and
hierarchical computation.  Finally, one might ask whether the main
lower bounds presented in this paper also hold when recomputation of
intermediate values is allowed.

As a broader consideration, our lower bound techniques, as well
as others in the literature, crucially exploit the circumstance that
the execution of some algorithms embeds the evaluation of the same
DAG for different inputs.  The development of communication lower
bound techniques for algorithms (e.g., heapsort or quicksort) which
do not fall in this class remains an open, challenging problem.

\paragraph{Acknowledgements.}
The authors would like to thank Lorenzo De Stefani, Andrea
Pietracaprina, and Geppino Pucci for insightful discussions.

\bibliographystyle{abbrv}
\bibliography{biblio.bib}

\begin{thebibliography}{10}

\bibitem{AggarwalCS90}
A.~Aggarwal, A.~K. Chandra, and M.~Snir.
\newblock Communication complexity of {PRAMs}.
\newblock {\em Theoret. Comput. Sci.}, 71(1):3--28, 1990.

\bibitem{AggarwalV88}
A.~Aggarwal and J.~S. Vitter.
\newblock The input/output complexity of sorting and related problems.
\newblock {\em Comm. ACM}, 31(9):1116--1127, 1988.

\bibitem{AjtaiKS83}
M.~Ajtai, J.~Koml{\'{o}}s, and E.~Szemer{\'{e}}di.
\newblock Sorting in $c \log n$ parallel steps.
\newblock {\em Combinatorica}, 3(1):1--19, 1983.

\bibitem{BallardDGLOST16}
G.~Ballard, J.~Demmel, A.~Gearhart, B.~Lipshitz, Y.~Oltchik, O.~Schwartz, and
  S.~Toledo.
\newblock Network topologies and inevitable contention.
\newblock In {\em Proceedings of the 1st International Workshop on
  Communication Optimizations in HPC (COMHPC)}, pages 39--52, 2016.

\bibitem{BallardDHS11}
G.~Ballard, J.~Demmel, O.~Holtz, and O.~Schwartz.
\newblock Minimizing communication in numerical linear algebra.
\newblock {\em SIAM J. Matrix Anal. Appl.}, 32(3):866--901, 2011.

\bibitem{BallardDHS12}
G.~Ballard, J.~Demmel, O.~Holtz, and O.~Schwartz.
\newblock Graph expansion and communication costs of fast matrix
  multiplication.
\newblock {\em J. ACM}, 59(6), 2012.

\bibitem{Batcher68}
K.~E. Batcher.
\newblock Sorting networks and their applications.
\newblock In {\em Proceedings of the AFIPS Spring Joint Computer Conference},
  volume~32, pages 307--314, 1968.

\bibitem{BaumkerDH98}
A.~B{\"a}umker, W.~Dittrich, and F.~{Meyer auf der Heide}.
\newblock Truly efficient parallel algorithms: 1-optimal multisearch for an
  extension of the {BSP} model.
\newblock {\em Theoret. Comput. Sci.}, 203(2):175--203, 1998.

\bibitem{Benes64}
V.~E. Bene\v{s}.
\newblock Permutation groups, complexes, and rearrangeable connecting networks.
\newblock {\em Bell System Tech. J.}, 43:1619--1640, 1964.

\bibitem{BhattBP08}
S.~N. Bhatt, G.~Bilardi, and G.~Pucci.
\newblock Area-time tradeoffs for universal {VLSI} circuits.
\newblock {\em Theoret. Comput. Sci.}, 408(2-3):143--150, 2008.

\bibitem{Bilardi89}
G.~Bilardi.
\newblock Merging and sorting networks with the topology of the omega network.
\newblock {\em IEEE Trans. Comput.}, 38(10):1396--1403, 1989.

\bibitem{BilardiF11}
G.~Bilardi and C.~Fantozzi.
\newblock New area-time lower bounds for the multidimensional {DFT}.
\newblock In {\em Proceedings of the 17th Computing: The Australasian Theory
  Symposium (CATS)}, pages 111--120, 2011.

\bibitem{BilardiPD00}
G.~Bilardi, A.~Pietracaprina, and P.~D'Alberto.
\newblock On the space and access complexity of computation {DAGs}.
\newblock In {\em Proceedings of the 26th International Workshop on
  Graph-Theoretic Concepts in Computer Science (WG)}, pages 47--58, 2000.

\bibitem{BilardiPP07}
G.~Bilardi, A.~Pietracaprina, and G.~Pucci.
\newblock Decomposable {BSP}: a bandwidth-latency model for parallel and
  hierarchical computation.
\newblock In {\em Handbook of Parallel Computing: Models, Algorithms and
  Applications}, pages 277--315. CRC Press, 2007.

\bibitem{BilardiPPSS16}
G.~Bilardi, A.~Pietracaprina, G.~Pucci, M.~Scquizzato, and F.~Silvestri.
\newblock Network-oblivious algorithms.
\newblock {\em J. {ACM}}, 63(1), 2016.

\bibitem{BilardiP86}
G.~Bilardi and F.~Preparata.
\newblock Area-time lower-bound techniques with applications to sorting.
\newblock {\em Algorithmica}, 1(1):65--91, 1986.

\bibitem{BilardiP99}
G.~Bilardi and F.~Preparata.
\newblock Processor-time tradeoffs under bounded-speed message propagation:
  Part {II}, lower bounds.
\newblock {\em Theory Comput. Syst.}, 32(5):531--559, 1999.

\bibitem{BilardiSS12}
G.~Bilardi, M.~Scquizzato, and F.~Silvestri.
\newblock A lower bound technique for communication on {BSP} with application
  to the {FFT}.
\newblock In {\em Proceedings of the 18th International European Conference on
  Parallel and Distributed Computing (Euro-Par)}, pages 676--687, 2012.

\bibitem{BlellochFGGS16}
G.~E. Blelloch, J.~T. Fineman, P.~B. Gibbons, Y.~Gu, and J.~Shun.
\newblock Efficient algorithms with asymmetric read and write costs.
\newblock In {\em Proceedings of the 24th Annual European Symposium on
  Algorithms (ESA)}, pages 14:1--14:18, 2016.

\bibitem{ChowdhuryRSB13}
R.~A. Chowdhury, V.~Ramachandran, F.~Silvestri, and B.~Blakeley.
\newblock Oblivious algorithms for multicores and networks of processors.
\newblock {\em J. Parallel Distrib. Comput.}, 73(7):911--925, 2013.

\bibitem{ColeR12}
R.~Cole and V.~Ramachandran.
\newblock Efficient resource oblivious algorithms for multicores with false
  sharing.
\newblock In {\em Proceedings of the 26th IEEE International Parallel and
  Distributed Processing Symposium (IPDPS)}, pages 201--214, 2012.

\bibitem{ColeR17}
R.~Cole and V.~Ramachandran.
\newblock Resource oblivious sorting on multicores.
\newblock {\em ACM Trans. Parallel Comput.}, 3(4), 2017.

\bibitem{CooleyT65}
J.~W. Cooley and J.~W. Tukey.
\newblock An algorithm for the machine calculation of complex fourier series.
\newblock {\em Math. Comput.}, 19:297--301, 1965.

\bibitem{CullerKPSSSSE96}
D.~E. Culler, R.~M. Karp, D.~A. Patterson, A.~Sahay, E.~E. Santos, K.~E.
  Schauser, R.~Subramonian, and T.~von Eicken.
\newblock {LogP}: A practical model of parallel computation.
\newblock {\em Comm. ACM}, 39(11):78--85, 1996.

\bibitem{delaTorreK96}
P.~de~la Torre and C.~P. Kruskal.
\newblock Submachine locality in the bulk synchronous setting.
\newblock In {\em Proceedings of the 2nd International Conference on Parallel
  Processing (Euro-Par)}, pages 352--358, 1996.

\bibitem{DowdPRS89}
M.~Dowd, Y.~Perl, L.~Rudolph, and M.~Saks.
\newblock The periodic balanced sorting network.
\newblock {\em J. ACM}, 36(4):738--757, 1989.

\bibitem{FrigoLPR12}
M.~Frigo, C.~E. Leiserson, H.~Prokop, and S.~Ramachandran.
\newblock Cache-oblivious algorithms.
\newblock {\em ACM Trans. Algorithms}, 8(1), 2012.

\bibitem{Goodrich99}
M.~T. Goodrich.
\newblock Communication-efficient parallel sorting.
\newblock {\em SIAM J. Comput.}, 29(2):416--432, 1999.

\bibitem{Hall35}
P.~Hall.
\newblock On representatives of subsets.
\newblock {\em J. London Math. Soc.}, 10(1):26--30, 1935.

\bibitem{Hennie65}
F.~C. Hennie.
\newblock One-tape, off-line {T}uring machine computations.
\newblock {\em Information and Control}, 8(6):553--578, 1965.

\bibitem{HongK81}
J.-W. Hong and H.~T. Kung.
\newblock {I/O} complexity: The red-blue pebble game.
\newblock In {\em Proceedings of the 13th Annual ACM Symposium on Theory of
  Computing (STOC)}, pages 326--333, 1981.

\bibitem{HopcroftPV77}
J.~E. Hopcroft, W.~J. Paul, and L.~G. Valiant.
\newblock On time versus space.
\newblock {\em J. {ACM}}, 24(2):332--337, 1977.

\bibitem{IronyTT04}
D.~Irony, S.~Toledo, and A.~Tiskin.
\newblock Communication lower bounds for distributed-memory matrix
  multiplication.
\newblock {\em J. Parallel Distrib. Comput.}, 64(9):1017--1026, 2004.

\bibitem{JuurlinkW98}
B.~H.~H. Juurlink and H.~A.~G. Wijshoff.
\newblock A quantitative comparison of parallel computation models.
\newblock {\em ACM Trans. Comput. Syst.}, 16(3):271--318, 1998.

\bibitem{Knuth73}
D.~E. Knuth.
\newblock {\em The Art of Computer Programming, Volume {III:} Sorting and
  Searching}.
\newblock Addison-Wesley, 2nd edition, 1973.

\bibitem{KochLMRRS97}
R.~R. Koch, F.~T. Leighton, B.~M. Maggs, S.~B. Rao, A.~L. Rosenberg, and E.~J.
  Schwabe.
\newblock Work-preserving emulations of fixed-connection networks.
\newblock {\em J. ACM}, 44(1):104--147, 1997.

\bibitem{Lawrie75}
D.~H. Lawrie.
\newblock Access and alignment of data in an array processor.
\newblock {\em IEEE Trans. Comput.}, C-24(12):1145--1155, 1975.

\bibitem{Leighton92}
F.~T. Leighton.
\newblock {\em Introduction to Parallel Algorithms and Architectures: Arrays,
  Trees, Hypercubes}.
\newblock Morgan Kaufmann Publishers Inc., 1992.

\bibitem{MacKenzieR98}
P.~D. MacKenzie and V.~Ramachandran.
\newblock Computational bounds for fundamental problems on general-purpose
  parallel models.
\newblock In {\em Proceedings of the 10th Annual ACM Symposium on Parallel
  Algorithms and Architectures (SPAA)}, pages 152--163, 1998.

\bibitem{PapadimitriouU87}
C.~H. Papadimitriou and J.~D. Ullman.
\newblock A communication-time tradeoff.
\newblock {\em SIAM J. Comput.}, 16(4):639--646, 1987.

\bibitem{PatersonH70}
M.~S. Paterson and C.~E. Hewitt.
\newblock Comparative schematology.
\newblock In {\em Proceedings of the Project MAC Conference on Concurrent
  Systems and Parallel Computation}, pages 119--127, 1970.

\bibitem{RanjanSZ11}
D.~Ranjan, J.~Savage, and M.~Zubair.
\newblock Strong {I/O} lower bounds for binomial and {FFT} computation graphs.
\newblock In {\em Proceedings of the 17th Annual International Conference on
  Computing and Combinatorics (COCOON)}, pages 134--145, 2011.

\bibitem{Savage95}
J.~E. Savage.
\newblock Extending the {H}ong-{K}ung model to memory hierarchies.
\newblock In {\em Proceedings of the 1st Annual International Conference on
  Computing and Combinatorics (COCOON)}, pages 270--281, 1995.

\bibitem{Savage98}
J.~E. Savage.
\newblock {\em Models of Computation: Exploring the Power of Computing}.
\newblock Addison-Wesley Longman Publishing Co., Inc., 1998.

\bibitem{ScquizzatoS14}
M.~Scquizzato and F.~Silvestri.
\newblock Communication lower bounds for distributed-memory computations.
\newblock In {\em Proceedings of the 31st Symposium on Theoretical Aspects of
  Computer Science (STACS)}, pages 627--638, 2014.

\bibitem{Thompson80}
D.~C. Thompson.
\newblock {\em A Complexity Theory for VLSI}.
\newblock PhD thesis, Carnegie Mellon University, 1980.

\bibitem{Tiskin98}
A.~Tiskin.
\newblock The bulk-synchronous parallel random access machine.
\newblock {\em Theoret. Comput. Sci.}, 196(1-2):109--130, 1998.

\bibitem{Tiskin11}
A.~Tiskin.
\newblock {BSP} (bulk synchronous parallelism).
\newblock In {\em Encyclopedia of Parallel Computing}, pages 192--199.
  Springer, 2011.

\bibitem{Valiant90}
L.~G. Valiant.
\newblock A bridging model for parallel computation.
\newblock {\em Comm. ACM}, 33(8):103--111, 1990.

\bibitem{Vitter06}
J.~S. Vitter.
\newblock Algorithms and data structures for external memory.
\newblock {\em Foundations and Trends in Theoretical Computer Science},
  2(4):305--474, 2006.

\bibitem{Vuillemin83}
J.~Vuillemin.
\newblock A combinatorial limit to the computing power of {VLSI} circuits.
\newblock {\em IEEE Trans. Comput.}, 32(3):294--300, 1983.

\bibitem{WuF81}
C.-L. Wu and T.-Y. Feng.
\newblock The universality of the shuffle-exchange network.
\newblock {\em IEEE Trans. Comput.}, 30(5):324--332, 1981.

\bibitem{Yao79}
A.~C.-C. Yao.
\newblock Some complexity questions related to distributive computing.
\newblock In {\em Proceedings of the 11th Annual ACM Symposium on Theory of
  Computing (STOC)}, pages 209--213, 1979.

\end{thebibliography}

\appendix
\section*{APPENDIX}
\setcounter{section}{1}

\subsection{An Improved Dominator Analysis for the FFT DAG}

Given a DAG, $D(k)$ denotes the maximum size of a set $U$ of nodes that has a dominator $W$ of size $k$.
Hong and Kung showed that, for the FFT DAG, $D(k) \leq 2k \log k$~\cite[Theorem~4.1]{HongK81}.
In this section we show that this bound can be improved to $D(k) \leq k \log 2k$, which
is tight for every $k$ power of two. This can be done by modifying the inductive proof
of Hong and Kung accordingly. In the proof we will use the following lemma.

\begin{lemma}\label{lem}
If $0 \leq x \leq y$ and $x+y \leq m$, then
\[
x \log x + y \log y + 2x \leq m \log m.
\]
\end{lemma}

\begin{proof}
Let $m' = x + y$. Standard calculus shows that, in the interval $0 \leq x \leq m'/2$,
\[
x \log x + (m'-x) \log (m'-x) + 2x \leq m' \log m'.
\]
Since, by hypothesis, $m' \leq m$, we obtain $x \log x + y \log y + 2x \leq m' \log m' \leq m \log m$.
\end{proof}

We are now ready to show the result claimed at the beginning of this section.

\begin{proposition}
For $k \geq 2$, any node set $U$ of the FFT DAG that has a dominator set of size
no more than $k$ can have at most $k \log 2k$ nodes.
\end{proposition}

\begin{proof}
The proof is by induction on $k$. Since $2k \log k = k \log 2k$ when $k=2$,
the base case is the same as in the proof of Hong and Kung.

We now mimic the inductive argument in the proof of Theorem~4.1 of~\cite{HongK81}.
We partition the nodes of the FFT DAG into three parts, $A$, $B$, and $C$, defined as follows.
\begin{align*}
A &= \{ \text{nodes}\ \langle w,l \rangle\ \text{s.t.}\ 0 \leq w < n/2\ \text{and}\ l < \log n \},\\
B &= \{ \text{nodes}\ \langle w,l \rangle\ \text{s.t.}\ n/2 \leq w < n\ \text{and}\ l < \log n \},\\
C &= \{ \text{nodes}\ \langle w,l \rangle\ \text{s.t.}\ l = \log n \}.
\end{align*}
The set of nodes $\langle w,l \rangle \in C$ such that $0 \leq w < n/2$ is said to be
the upper half of part $C$, whereas the set of nodes $\langle w,l \rangle \in C$ such that
$n/2 \leq w < n$ is said to be the lower half of part $C$. (See also the figure depicted
in the proof of Theorem~4.1 of~\cite{HongK81}.) The dominator is partitioned
into three parts, $D_A$, $D_B$, and $D_C$, which have $d_A$, $d_B$, and $d_C$
nodes respectively. Without loss of generality we assume $d_A \leq d_B$. The set $U$
is partitioned into three parts, $U_A$, $U_B$, and $U_C$, which have $u_A$, $u_B$,
and $u_C$ nodes respectively. If $u_C > d_C + 2d_A$ then either there are more
than $d_A$ nodes of $U_C \setminus D_C$ in the upper half of part $C$ or there
are more than $d_A$ nodes of $U_C \setminus D_C$ in the lower half of part $C$.
In either case, there are more than $d_A$ independent paths from the upper half
inputs (i.e., the set of input nodes of the FFT DAG such that $0 \leq w < n/2$)
to these nodes in $U_C \setminus D_C$. Since the set $D_A$ has only $d_A$
nodes, this is impossible. Therefore we have
\[
u_C \leq d_C + 2d_A.
\]

By inductive hypothesis we have
\begin{align*}
u_A &\leq d_A \log 2d_A, \\
u_B &\leq d_B \log 2d_B.
\end{align*}
Thus
\[
|U| \leq d_A \log 2d_A + d_B \log 2d_B + d_C + 2d_A.
\]
Combining Lemma~\ref{lem} with the hypothesis $x+y \leq m$ yields
\begin{equation}\label{eq:lemma}
x \log 2x + y \log 2y + 2x \leq m \log 2m.
\end{equation}
Since $0 \leq d_A \leq d_B$ and $d_A + d_B \leq k - d_C$,
applying~\eqref{eq:lemma} yields
\[
|U| \leq (k-d_C) \log 2(k-d_C) + d_C \leq k \log 2k,
\]
as desired.
\end{proof}

\end{document}